\documentclass[11pt]{article}
\usepackage{graphicx,epsfig}
\usepackage{color}
\usepackage{amsmath,amssymb,mathrsfs,amsthm}
\usepackage{amsfonts}
\usepackage{multicol}
\usepackage{graphicx}
\usepackage{color}
\usepackage{extarrows}
\usepackage{pdfcolmk}
\usepackage[nice]{nicefrac}
\usepackage{amsthm}
\usepackage{latexsym}
\usepackage{dsfont}
\usepackage{setspace}
\usepackage[a4paper, tmargin=1.5in, bmargin=1in, lmargin=1in, rmargin=1in, headheight=13.6pt]{geometry}
\usepackage{hyperref}
\usepackage[sort&compress,numbers]{natbib}
\usepackage[fit]{truncate}
\usepackage[font=small,labelfont=bf,labelsep=space]{caption}
\usepackage{algorithm}
\usepackage[noend]{algpseudocode}
\usepackage{caption}
\mathchardef\mhyphen="2D
\algrenewcommand{\algorithmicrequire}{\textbf{Input:}}
\algrenewcommand{\algorithmicensure}{\textbf{Output:}}
\makeatletter

\date{}

\def\@citex[#1]#2{\if@filesw\immediate\write\@auxout{\string\citation{#2}}\fi
  \def\@citea{}\@cite{\@for\@citeb:=#2\do
    {\@citea\def\@citea{,\linebreak[0]\hskip0pt plus .2em}%
      \@ifundefined{b@\@citeb}%
    {{\bf ?}\@warning{Citation `\@citeb' on page \thepage\space undefined}}%
      \hbox{\csname b@\@citeb\endcsname}}}{#1}}
\newtheorem{theorem}{Theorem}[section]

\newtheorem{definition}[theorem]{Definition}

\newtheorem{rule-def}[theorem]{Rule}
\numberwithin{equation}{section} \makeatletter
\usepackage{algorithm}
\usepackage[noend]{algpseudocode}
\usepackage{caption}
\algrenewcommand{\algorithmicrequire}{\textbf{Input:}}
\algrenewcommand{\algorithmicensure}{\textbf{Output:}}
\newsavebox{\savepar}
\newenvironment{boxit}{\begin{lrbox}{\savepar}
\begin{minipage}[b]{6.2in}} {\end{minipage}\end{lrbox}\fbox{\usebox{\savepar}}}
\begin{document}
\title{UC Secure Issuer-Free Adaptive Oblivious Transfer with Hidden Access Policy}
\author{Vandana Guleria\thanks{Corresponding
author. E-mail:
vandana.math@gmail.com},~ Ratna Dutta\\
$\rm $\small {Department of Mathematics} \\
\small {Indian Institute of Technology Kharagpur}\\
\small {Kharagpur -721302, India}}
\maketitle{}
\begin{abstract}
Privacy is a major concern in designing any cryptographic primitive when frequent transactions are done electronically. During electronic transactions, people reveal their personal data into several servers and believe that this information does not leak too much about them. The\emph{ adaptive oblivious transfer with hidden access policy (${\sf AOT\mhyphen HAP}$)} takes measure against such privacy issues. The existing ${\sf AOT\mhyphen HAP}$ involves a sender and multiple receivers apart from a designated issuer. Security of these schemes rely on the fact that the issuer cannot collude with a set of receivers. Moreover, they loose security when run with multiple protocol instances during concurrent execution. We present the \emph{first issuer-free ${\sf AOT\mhyphen HAP}$} in \emph{universal composable (UC) framework} in which the protocol is secure even when composed with each other or with other protocols. A concrete security analysis is given assuming the hardness of $q$-strong Diffie-Hellman (SDH), decision Linear (DLIN) and decision bilinear Diffie-Hellman (DBDH) problems against malicious adversary in UC model. Moreover, the protocol outperforms the existing similar schemes.\\\\

\end{abstract}

\section{Introduction}
To protect sensitive data, it is required to associate some access policies to the data so that receivers who have the necessary permissions can access it. Receiver should possess attribute sets in order to retrieve the data. Upon controlling access to the sensitive data, the associated access policy leaks too much information. For instance, consider a scenario of a medical database consisting of patients records. The access policy associated with each patient's record includes patient's identity, names of treating doctors, names of diseases, medicines recommended, treatments done etc. If access policies are public, then they reveal about the patient's personal information. Many patients do not want to disclose their identities when they undergo a plastic surgical operation, a psychiatric treatment etc. Sometimes, even doctors want to remain anonymous while treating a famous person. The solution to such privacy problems is \emph{adaptive oblivious transfer realizing hidden access policy (${\sf AOT\mhyphen HAP}$)}.

\noindent The ${\sf AOT\mhyphen HAP}$ involves a sender, multiple receivers and an issuer. The sender carries a database with $N$ messages and generates a ciphertext database. The access policies are embedded in ciphertexts and are kept veiled. Only those receivers who have the necessary attributes can decipher the ciphertext using secret key issued by the issuer corresponding to its attribute set. A receiver decode the message correctly if his attribute set satisfies the access structure embedded in the ciphertext. Otherwise, a wrong message is retrieved by the receiver. In an adaptive $k\mhyphen {\rm out}\mhyphen {\rm of}\mhyphen N$ oblivious transfer, each receiver can recover only $k$ messages and remains oblivious to other $N-k$ messages. The term adaptive means, a receiver recovers one message in each of $k$ transactions. The ${\sf AOT\mhyphen HAP}$ consists of three phases -- initialization, issue and transfer. The ciphertext database is generated at initialization phase. In issue phase, the secret key is generated corresponding to a receiver's attribute set. Using secret key, a receiver contacts with the sender to decode the ciphertext in transfer phase. Initialization phase is executed once at the outset of the protocol. Issue phase is also executed once between the issuer and a receiver. On the other hand, transfer phase is run $k$ times between a receiver and the sender.

\noindent  Rabin \cite{rabin1981exchange} introduced the first oblivious transfer which was subsequently followed by Brassard {\em et al.} \cite{brassard1987all}  and Even {\em et al.} \cite{even1985randomized}
to construct secure
protocols for multiparty computation. Tauman
\cite{kalai2005smooth} proposed an efficient contruction of oblivious transfer protocol
using projective hash framework of Cramer and Shoup
\cite{cramer2002universal}.
Naor and Pinkas \cite{naor2001efficient} generalized the concept of Tauman.  A wide variety of oblivious transfer protocols are available in the literature
\cite{aiello2001priced,
kushilevitz1997replication, tassa2011generalized, lindell2008efficient, peikert2008framework, damgaard2009essentially, bresson2003simple, abdallasphf}.

\noindent The first ${\sf AOT}$ was introduced by Naor and Pinkas \cite{naor1999oblivious} consequently followed by \cite{camenisch2007simulatable}, \cite{green2007blind}. Aforementioned ${\sf AOT}$ do not restrict the receivers to access the content of the database. Coull {\em et al.} \cite{coull2009controlling} presented the first ${\sf AOT}$ with access policy which was subsequently generalized by Camenisch {\em et al.} \cite{camenisch2009oblivious} and Zhang {\em et
al.} \cite{zhang2010oblivious}.  Recently, Guleria and Dutta \cite{guleria2015universally} presented the first issuer-free ${\sf AOT}$ with access policy in universal composable (UC) framework. However, the access policies are public in
\cite{camenisch2009oblivious}, \cite{coull2009controlling}, \cite{guleria2015universally} and
\cite{zhang2010oblivious}. To overcome this, Camenisch {\em et al,}  \cite{camenisch2012oblivious}, \cite{camenisch2011oblivious} and Guleria and Dutta \cite{guleria2014adaptive} introduced {\sf AOT-HAP} which are
the only oblivious transfer protocols
with hidden access policies. The security framework adapted in \cite{naor1999oblivious} is half-simulation model while that in \cite{camenisch2007simulatable}, \cite{green2007blind}, \cite{coull2009controlling}, \cite{camenisch2009oblivious}, \cite{zhang2010oblivious},  \cite{camenisch2012oblivious}, \cite{camenisch2011oblivious} and \cite{guleria2014adaptive} is full-simulation model. In the half-simulation model, the security of one party is examined by comparing the output of the real world against the turnout of the ideal world, while the security of the other party is examined by an argument. The full-simulation model addresses the security of the sender and a receiver in real/ideal world paradigm. None of these models support concurrent execution. Abe {\em et al.} \cite{abe2013universally} proposed the first UC secure ${\sf AOT}$ with access policy. The UC secure framework follows real/ideal world paradigm and also remains secure even when composed with each other or with the instances of other protocols.

\noindent \textbf{Our Contribution.} The aforesaid ${\sf AOT\mhyphen HAP}$ protocols do not consider collusion of an issuer with a collection of receivers. If the issuer colludes with any receiver, the receiver can decipher the ciphertext database by getting private keys for all the attributes in the universe of attributes, hence, making these protocols insecure. Besides, the existing ${\sf AOT\mhyphen HAP}$ protocols \cite{camenisch2012oblivious}, \cite{camenisch2011oblivious}, \cite{guleria2014adaptive} are secure in full-simulation model which are vulnerable to practical attacks against concurrent execution. Designs without such constraints are more desirable in practice from both security and efficiency point of view. In this paper, we concentrate on designing ${\sf AOT\mhyphen HAP}$ by eliminating these limitations. Our contribution in this paper is--

\begin{itemize}
\vspace{-.2cm}
\item firstly to contruct an issuer-free ${\sf AOT\mhyphen HAP}$,\vspace{-.2cm}
\item secondly to prove security in UC framework so that it does not loose security when composed with each other or with other protocols.
\vspace{-.2cm}
\end{itemize}

\noindent Our issuer-free $k\mhyphen {\rm out}\mhyphen {\rm of}\mhyphen N$ ${\sf AOT\mhyphen HAP}$ is executed between a sender and multiple receivers. We combine Boneh-Boyen (BB) \cite{boneh2004shortsignature} signature and Nishide {\em et al.} \cite{nishide2008attribute} ciphertext policy attribute based encryption (CP-ABE). Groth-Sahai proofs \cite{groth2008efficient} are employed for non-interactive verification of pairing product equations. The sender holds a database of $N$ messages, and each receiver carries some attributes. Each message is associated with an access policy. The sender generates ciphertext database in initialization phase by encrypting each message using CP-ABE under its access policy. The ciphertext database is made public, but the access policies are kept hidden. Subsequently, the protocol makes a motion to issue phase, in which a receiver interacts with the sender to get private key corresponding to its attribute set without leaking any information about its attributes to the sender. Finally, the receiver invokes transfer phase to recover the message. The receiver can retrieve at most $k$ messages, one in each transfer phase in our $k\mhyphen {\rm out}\mhyphen {\rm of}\mhyphen N$ ${\sf AOT\mhyphen HAP}$.

\noindent The security of our ${\sf AOT\mhyphen HAP}$ is analyzed in UC framework in static corruption model in which, the adversary pre-decides which party to corrupt at the beginning of the protocol. This allows security against individual malicious receiver as well as malicious sender. The malicious adversary do not follow the protocol specifications and its malicious behaviour is controlled by providing non-interactive zero-knowledge \cite{groth2008efficient} proofs. We use Groth-Sahai non-interactive witness indistinguishability (${\sf NIWI}$) and non-interactive zero-knowledge (${\sf NIZK}$) proofs to detect the malicious activities of a receiver and the sender, respectively. The proposed ${\sf AOT\mhyphen HAP}$ provides the following security guarantees --

\begin{itemize}
\vspace{-.2cm}
\item the sender does not know who queries a message, which message is being queried and what is the access policy of the message,\vspace{-.2cm}
\item the receiver learns one message per query,\vspace{-.2cm}
\item the receiver recovers the message correctly if its attribute set qualifies the access policy associated with the message,\vspace{-.2cm}
\item the receiver does not learn anything about the access policy of the message.\vspace{-.2cm}
\end{itemize}

\noindent The proposed ${\sf AOT\mhyphen HAP}$ is UC secure assuming $q$-strong Diffie-Hellman (SDH), decision Linear (DLIN) and decision bilinear Diffie-Hellman (DBDH) assumptions in the presence of malicious adversary.

\noindent We compare the proposed ${\sf AOT\mhyphen HAP}$ with the existing similar schemes, \cite{camenisch2012oblivious}, \cite{camenisch2011oblivious} and \cite{guleria2014adaptive}, which are the only schemes available in literature. The ${\sf AOT\mhyphen HAP}$ \cite{camenisch2012oblivious}, \cite{camenisch2011oblivious} and \cite{guleria2014adaptive} assume an issuer and are secure in full-simulation model. The ${\sf AOT}$ scheme of \cite{abe2013universally} is secure in UC framework. However, it invokes an issuer and access policies are public. In contrast, our proposed ${\sf AOT\mhyphen HAP}$ relaxing the need of an issuer and is secure in UC framework. In \cite{abe2013universally}, \cite{camenisch2012oblivious}, \cite{camenisch2011oblivious} and \cite{guleria2014adaptive}, an issuer is included to certify whether each of the receiver possesses the claimed attributes. In such protocols, when a receiver sends an attribute to the ideal functionality, the functionality passes the attribute and receiver identity to the issuer. The issuer informs the functionality whether the attribute should be rejected or accepted. In our protocol, when a receiver sends an attribute to the ideal functionality, the functionality passes the randomized attribute to the sender. The sender checks whether the randomized attribute satisfies the associated pairing product equation. If yes, the sender tells the functionality that the attribute should be accepted, otherwise, rejected. The access policies associated with the messages are completely hidden. The encryption of the messages is such that only those receivers whose attribute set satisfies the access policy attached with messages could recover a correct message otherwise a random message is recovered.

\section{Preliminaries}\label{preliminaries}
\noindent \textbf{Notations:} Throughout, $\rho$ is taken as the security parameter, $x \xleftarrow{\$} A$
is a random element sampled from the set $A$, $y \leftarrow B$
indicates algorithm $B$ outputs y, $X \stackrel{c}{\approx} Y$ denotes $X$ and $Y$ are indistinguishable computationally and $\mathbb{N}$ denotes natural numbers. 

\begin{definition}(Bilinear Pairing)
The map
$e : \mathbb{G}_{1} \times \mathbb{G}_{2} \rightarrow
\mathbb{G}_T$ is \emph{bilinear} if the following
conditions are stisfied. (i) $e(x^{a}, y^{b}) = e(x, y)^{ab} ~\forall~ x \in
\mathbb{G}_1,  y \in \mathbb{G}_2, a, b \in \mathbb{Z}_{p}$.
(ii) $e(x, y)$ generates $\mathbb{G}_T$,
$~\forall~ x \in \mathbb{G}_1,  y \in
\mathbb{G}_2, x \neq 1,  y \neq 1$.
(iii) $e(x, y)$ is efficiently computable
$~\forall~ x \in \mathbb{G}_1, y \in
\mathbb{G}_2$, where $\mathbb{G}_1,
\mathbb{G}_2$ and $\mathbb{G}_T$ are three multiplicative cyclic
groups of prime order $p$ and $g_1$ and $g_2$ are generators of
$\mathbb{G}_1$ and $\mathbb{G}_2$ respectively.
\end{definition}
\noindent If $\mathbb{G}_1 \neq \mathbb{G}_2$, then $e$ is {\em
asymmetric} bilinear pairing. Otherwise, $e$ is {\em symmetric}
bilinear pairing.

\begin{definition}($q$-SDH \cite{boneh2004shortsignature}) The advantage ${\sf
Adv}^{q-SDH}_{\mathbb{G}}(\mathcal{A}) = {\sf Pr}[\mathcal{A}(g,
g^x$, $g^{x^2}, \ldots, g^{x^q}) = (c$, $g^{\frac{1}{x + c}})]$ of $q$-Strong Diffie-Hellman (SDH)
assumption is
negligible in $\rho$, where $g \xleftarrow{\$} \mathbb{G}, x
\xleftarrow{\$} \mathbb{Z}_{p}, c \in \mathbb{Z}_{p}$ and $\mathcal{A}$ is a PPT algorithm with running time
in $\rho$.
\end{definition}

\begin{definition}(DLIN \cite{boneh2004short}) The advantage ${\sf
Adv}^{DLIN}_{\mathbb{G}}(\mathcal{A}) = {\sf Pr}[\mathcal{A}(g,
g^a, g^b$, $g^{ra}, g^{sb}, g^{r+s})] - {\sf Pr}[\mathcal{A}(g, g^a,
g^b, g^{ra}, g^{sb}, t)]$ of Decision Linear (DLIN) assumption in
is negligible in $\rho$, where $g
\xleftarrow{\$} \mathbb{G}, t \xleftarrow{\$} \mathbb{G}, a, b, r,
s \in \mathbb{Z}_{p}$ and $\mathcal{A}$ is a PPT algorithm with running time
in $\rho$.
\end{definition}

\begin{definition}(DBDH \cite{waters2005efficient}) The advantage ${\sf Adv}^{DBDH}_{\mathbb{G},
\mathbb{G}_T}(\mathcal{A}) = {\sf Pr}[\mathcal{A}(g, g^a, g^b,
g^c, e(g, g)^{abc})] - {\sf Pr}[\mathcal{A}(g, g^a$, $g^b, g^c,
Z)]$ of Decision Bilinear
Diffie-Hellman (DBDH) assumption is negligible in $\rho$, where $g \xleftarrow{\$} \mathbb{G},
Z \xleftarrow{\$} \mathbb{G}_T, a, b, c \in \mathbb{Z}_{p}$ and $\mathcal{A}$ is a PPT algorithm with running time
in $\rho$.
\end{definition}

\noindent\begin{boxit}
\noindent \textbf{${\sf BilinearSetup}$: } The ${\sf BilinearSetup}$ on input
 $\rho$ outputs ${\sf params} = (p, \mathbb{G}_1, \mathbb{G}_2, \mathbb{G}_{T}, e, g_1, g_2)$, where $e :
\mathbb{G}_1 \times \mathbb{G}_2 \rightarrow \mathbb{G}_T$ is an asymmetric bilinear pairing, $g_1$
generates a group $\mathbb{G}_1$, $g_2$
generates a group $\mathbb{G}_2$ and $p$, the order of the groups
$\mathbb{G}_1$, $\mathbb{G}_2$ and $\mathbb{G}_T$, is prime, i.e ${\sf params} \leftarrow {\sf BilinearSetup}(1^{\rho})$.
\end{boxit}

\subsection{Non-Interactive Verification of Pairing Product Equation \cite{groth2008efficient}}\label{noninteractive}
A prover and a verifier runs the Groth-Sahai proofs for
non-interactive verification of a pairing product equation
\begin{eqnarray}\label{eq3}
\prod_{q=1}^{Q}e(a_q\prod_{i=1}^{n}x_{i}^{\alpha_{q, i}}, b_q\prod_{i=1}^{n}y_{i}^{\beta_{q, i}}) &=&t_{T},
\end{eqnarray}
where $a_{q} \in \mathbb{G}_1, b_{q} \in \mathbb{G}_2$, $\alpha_{q, i}, \beta_{q, i} \in \mathbb{Z}_{p}$ and $t_{T} \in
\mathbb{G}_{T}$, the coefficients of the pairing product
equation \ref{eq3} are given to the verifier, $q =1, 2,
\ldots, Q, i =1, 2, \ldots, n$.
The prover knows the secret values (also called witnesses) $x_{i =1, 2,
\ldots, n} \in \mathbb{G}_1, y_{i =1, 2, \ldots, n} \in \mathbb{G}_2$ that satisfy
the equation \ref{eq3}. The prover proves in a non-interactive way that he knows $x_{i}$ and $y_{i}$ without
giving any knowledge about secret values to the verifier. Let
$\mathcal{W} = \{x_{i =1, 2, \ldots, n}, y_{i =1, 2, \ldots, n}\}$
denotes the set of all secret values equation
\ref{eq3}. The set $\mathcal{W}$ is called witnesses.

%The product of two vectors is defined
%component wise, i.e, $(a_1, a_2, a_3)(b_1, b_2, b_3) = (a_1b_1,
%a_2b_2, a_3b_3)$ for $(a_1, a_2, a_3), (b_1, b_2, b_3) \in
%\mathbb{G}_{1}^3$ for a finite order group $\mathbb{G}_1$.

\noindent For non-interactive verification of the pairing product
equation \ref{eq3}, a common reference string ${\sf GS}$ is generated upon input a security
parameter $\rho$ as follows. Let ${\sf params} =
(p, \mathbb{G}_1, \mathbb{G}_2, \mathbb{G}_{T}, e, g_1$, $g_2)$ ~$\leftarrow {\sf
BilinearSetup}(1^\rho)$, set $u_1 = (g_{1}^a$, $1, g_1) \in \mathbb{G}_{1}^{3},
u_2 = (1, g_{1}^b, g_1) \in \mathbb{G}_{1}^{3}$, $u_3 =
u_{1}^{\xi_{1}}u_{2}^{\xi_{2}} = (g_{1}^{a\xi_{1}},
g_{1}^{b\xi_{2}}$, $g_{1}^{\xi_{1} + \xi_{2}}) \in \mathbb{G}_{1}^{3}$, $v_1 = (g_{2}^a, 1, g_2) \in \mathbb{G}_{2}^{3},
v_2 = (1, g_{2}^b, g_2) \in \mathbb{G}_{2}^{3}$, $v_3 =
v_{1}^{\xi_{1}}v_{2}^{\xi_{2}} = (g_{2}^{a\xi_{1}},
g_{2}^{b\xi_{2}}$, $g_{2}^{\xi_{1} + \xi_{2}}) \in \mathbb{G}_{2}^{3}$,
where $\xi_{1}, \xi_{2} \xleftarrow{\$} \mathbb{Z}_{p}$, $a, b \xleftarrow{\$} \mathbb{Z}_{p}$ and $\mu_1 : \mathbb{G}_1
\rightarrow \mathbb{G}_{1}^3$, $\mu_2 : \mathbb{G}_2
\rightarrow \mathbb{G}_{2}^3$, $\mu_T : \mathbb{G}_T \rightarrow
\mathbb{G}_{T}^9$ are three efficiently computable embeddings such
that
$\mu_1(g_1) = (1, 1, g_1), \mu_2(g_2) = (1, 1, g_2),  {\rm ~and~} \mu_T(t_T) = \left(
                                             \begin{array}{ccc}
                                               1 & 1 & 1 \\
                                               1 & 1 & 1 \\
                                               1 & 1 & t_T \\
                                             \end{array}
                                           \right) \forall~g_1 \in \mathbb{G}_1, g_2 \in \mathbb{G}_2, t_T \in \mathbb{G}_T.
$
Note that $\mu_T(t_T) \in\mathbb{G}_{T}^9$ written in matrix form for convenience. Two elements of $\mathbb{G}_{T}^9$ are multiplied component wise. The common reference string ${\sf GS} = (u_1, u_2, u_3, v_1, v_2, v_3)$ is made public to the prover and the verifier. The commitments of $x_{i =1, 2,
\ldots, n}$ and $y_{i =1, 2, \ldots, n}$ are generated
using ${\sf GS}$ by  the prover. To
commit $x_i \in \mathbb{G}_1$ and $y_i \in \mathbb{G}_2$, the prover
takes $r_{1i}, r_{2i}, r_{3i} \xleftarrow{\$} \mathbb{Z}_{p}$ and
$s_{1i}, s_{2i}, s_{3i} \xleftarrow{\$} \mathbb{Z}_{p}$, sets
$$
c_i = {\sf Com}(x_i) = \mu_1(x_i)u_{1}^{r_{1i}}u_{2}^{r_{2i}}u_{3}^{r_{3i}},~~
d_i = {\sf Com'}(y_i) = \mu_2(y_i)v_{1}^{s_{1i}}v_{2}^{s_{2i}}v_{3}^{s_{3i}},
$$
\noindent where ${\sf Com}(x_i) \in \mathbb{G}_{1}^{3}, {\sf Com'}(y_i) \in \mathbb{G}_{2}^{3}$ for each $i=1, 2, \ldots, n$. The following proof components are generated by the prover
$$
P_j = \prod_{q=1}^{Q}\left(\widehat{d_{q}}\right)^{\sum_{i=1}^{n}\alpha_{q, i}r_{ji}},~~
P'_j =  \prod_{q=1}^{Q}\left(\mu_1(a_q)\prod_{i=1}^{n}\mu_1(x_i)^{\alpha_{q, i}}\right)^{\sum_{i=1}^{n}\beta_{q, i}s_{ji}}$$
using random values $r_{ji}, s_{ji}$, which were used for generating commitments to $x_{i}, y_{i}$, and gives proof $\pi = (c_1, c_2, \ldots, c_n, d_1$, $d_2, \ldots, d_n, P_1, P_2, P_3, P'_1, P'_2, P'_3)$ to the verifier, where $$\widehat{d_q} = \mu_2(b_q)\prod_{i=1}^{n}d_{i}^{\beta_{q, i}}, ~~1\leq i \leq n, ~~j = 1, 2, 3.$$
The verifier computes
$$
\widehat{c_q} = \mu_1(a_q)\prod_{i=1}^{n}c_{i}^{\alpha_{q, i}},
\widehat{d_q} = \mu_2(b_q)\prod_{i=1}^{n}d_{i}^{\beta_{q, i}},
$$
using $c_i, d_i$, coefficients $\alpha_{q, i}, \beta_{q, i}$ and outputs {\sf VALID} if the following equation holds
\begin{eqnarray}\label{eq1}
\prod_{q=1}^{Q}F(\widehat{c_q}, \widehat{d_q}) = \mu_{T}(t_T)\prod_{j=1}^{3}F(u_j, P_j)F(P'_j, v_j),
\end{eqnarray}
where $F : \mathbb{G}_{1}^3 \times \mathbb{G}_{2}^3 \rightarrow \mathbb{G}_{T}^9$ is defined as
$$F((x_1, x_2, x_3), (y_1, y_2, y_3)) =
 \left(
   \begin{array}{ccc}
     e(x_1, y_1) & e(x_1, y_2) & e(x_1, y_3) \\
     e(x_2, y_1) & e(x_2, y_2) & e(x_2, y_3) \\
     e(x_3, y_1) & e(x_3, y_2) & e(x_3, y_3) \\
   \end{array}
 \right).$$
%Note that the function $F$ is also bilinear and $F((x_1, x_2, x_3)$, $(y_1, y_2, y_3))$ is an element of $\mathbb{G}_{T}^9$.
%The product of two elements of $\mathbb{G}_{T}^9$ is component wise. For convenience, it has been written in matrix form.

\noindent The equations \ref{eq1} and \ref{eq3} holds simultaneously. 
%The equation \ref{eq3} is non-linear. If in equation \ref{eq3} only $x_{i}$ or $y_{i}$ are secrets, then it is a linear equation. 
For a linear equation in which $y_{i}$ are secrets, the verifier verifies the following equation
\begin{align}
&\prod_{q=1}^{Q}F\left(\mu_1(a_q)\prod_{i=1}^{n}\mu_1(x_i)^{\alpha_{q, i}}, \widehat{d_q}\right) = \mu_{T}(t_T)\prod_{j=1}^{3}F(P'_j, v_j),\\
&{\rm where~~} P'_j = \prod_{q=1}^{Q}\left(\mu_1(a_q)\prod_{i=1}^{n}\mu_1(x_i)^{\alpha_{q, i}}\right)^{\sum_{i=1}^{n}\beta_{q, i}s_{ji}}, j = 1, 2, 3.~~
\end{align}
If $x_{i}$ are secrets, the verifier has to verify the following equation
\begin{align}
&\prod_{q=1}^{Q}F\left(\widehat{c_q}, \mu_2(b_q)\prod_{i=1}^{n}\mu_2(y_i)^{\beta_{q, i}}\right) = \mu_{T}(t_T)\prod_{j=1}^{3}F(u_j, P_j),\\
&{\rm where~~} P_j = \prod_{q=1}^{Q}\left(\mu_2(b_q)\prod_{i=1}^{n}\mu_2(y_i)^{\beta_{q, i}}\right)^{\sum_{i=1}^{n}\alpha_{q, i}r_{ji}}, j = 1, 2, 3.~~
\end{align}

\begin{theorem}\label{crs}\cite{groth2008efficient}
The common reference strings in perfectly sound setting and witness indistinguishability setting under DLIN assumption are computationally indistinguishable .
\end{theorem}

\begin{definition}(${\sf NIWI}$) The advantage $${\sf Adv}^{{\sf NIWI}}_{\mathbb{G}, \mathbb{G}_T}(\mathcal{A}) = {\sf Pr}\left[\mathcal{A}({\sf GS}, \mathcal{S}, \mathcal{W}_0) = \pi\right] - {\sf Pr}\left[\mathcal{A}({\sf GS}, \mathcal{S}, \mathcal{W}_1) = \pi\right]$$ of non-interactive witness- indistinguishable (${\sf NIWI}$) proof is negligible in $\rho$ under DLIN assumption, where ${\sf GS}$ denotes perfectly sound setting's common reference string, $\mathcal{S}$ is a pairing product equation, $\mathcal{W}_0, \mathcal{W}_1$ denote two different set of witnesses satisfying equation $\mathcal{S}$, $\pi$ denotes proof for equation $\mathcal{S}$ and $\mathcal{A}$ is a PPT algorithm with running time
in $\rho$.
\end{definition}

\begin{definition}(${\sf NIZK}$) The advantage $${\sf Adv}^{{\sf NIZK}}_{\mathbb{G}, \mathbb{G}_T}(\mathcal{A}) = {\sf Pr}[\mathcal{A}({\sf GS'}, \mathcal{S}, \mathcal{W}) = \pi_0] - {\sf Pr}[\mathcal{A}({\sf GS'}, \mathcal{S}, {\sf t_{sim}}) = \pi_1]$$ of non-interactive zero-knowledge (${\sf NIZK}$) proof is negligible in $\rho$ under DLIN assumption, where ${\sf GS'}$ denotes witness indistinguishability setting's common reference string, $\mathcal{S}$ is a pairing product equation, $\mathcal{W}$ denotes witnesses satisfying pairing equation $\mathcal{S}$, $\pi_0$ denotes proof for equation $\mathcal{S}$, $\pi_1$ denotes simulated proof for equation $\mathcal{S}$ and $\mathcal{A}$ is a PPT algorithm with running time
in $\rho$.
\end{definition}
%\vspace{-.3cm}
%\noindent The notations ${\sf NIWI}\{\{(x_{i}, y_{i}\}_{1\leq i \leq n}) | \prod_{q=1}^{Q}e(a_q$ $\prod_{i=1}^{n}x_{i}^{\alpha_{q, i}}$, $b_q\prod_{i=1}^{n}y_{i}^{\beta_{q, i}}) =t_{T}\}$ for ${\sf NIWI}$ proof and ${\sf NIZK}\{(x_{i}, y_{i}\}_{1\leq i \leq n}) | $ $\prod_{q=1}^{Q}e(a_q\prod_{i=1}^{n}x_{i}^{\alpha_{q, i}}, b_q\prod_{i=1}^{n}y_{i}^{\beta_{q, i}}) = t_{T} \}$ for ${\sf NIZK}$ proof are followed in our construction. The convention is that the
%quantities in the parenthesis denote elements the knowledge of
%which are being proved to the verifier by the prover while all
%other parameters are known to the verifier. We have the following theorem.

\begin{theorem}\label{niwi}\cite{groth2008efficient}
Under DLIN assumption, the Groth-Sahai proofs are composable ${\sf NIWI}$ and ${\sf NIZK}$ for
satisfiability of a set of pairing product equation under a bilinear group .
\end{theorem}

\subsection{Access Structure for Ciphertext \cite{nishide2008attribute}}\label{access structure}
Let $n$ be the total number of attributes, and the attributes are indexed as $\{\mathbb{A}_1, \mathbb{A}_2, \ldots, \mathbb{A}_n\}$. Each attribute $\mathbb{A}_{\ell}$ can take $n_{\ell}$ values, where $\ell=1, 2, \ldots, n$. Let $S_{\ell} = \{v_{\ell,1}, v_{\ell,2}, \ldots, v_{\ell, t}$, $\ldots, v_{\ell, n_{\ell}}\}$ be the set of possible values for $\mathbb{A}_{\ell}$. The notation $L = [L_1, L_2, \ldots, L_n]$ is used to denote attribute list or attribute set for a receiver $R$, where $L_{\ell} \in S_{\ell}$. Each receiver has exactly one attribute value from each $S_{\ell}$, $\ell=1, 2, \ldots, n$. The notation $W = [W_1, W_2, \ldots, W_n]$ is used to specify the access policy associated with a message, where $W_{\ell} \subseteq S_{\ell}$. The attribute list $L$ satisfies $W$ if and only if $L_{\ell} \in W_{\ell}$, for $\ell=1, 2, \ldots, n$. The notation $L\models W$ is used to denote $L$ satisfies $W$.
For instance, let $n=3$ and $m'$ be a message with
access policy $W = [W_1 = \{v_{1, 1} , v_{1, 3}\}, W_2 = \{v_{2, 2}\}, W_3 = \{v_{3, 1}, v_{3, 2}, v_{3, 3}\}]$. Consider two attribute lists $L = [v_{1, 1}, v_{2, 2}, v_{3, 1}], \widetilde{L} = [v_{1, 2}, v_{2, 2}, v_{3, 3}]$. The attribute list $L\models W$, but $\widetilde{L}$ does not satisfy $W$ as $v_{1, 2} \notin W_1$.

\section{Security Model}\label{securitymodel}
\subsection{Syntactic of ${\sf AOT\mhyphen HAP}$}
The adaptive oblivious transfer with hidden access policy (${\sf AOT\mhyphen HAP}$) is a tuple of PPT algorithms and interactive protocols, ${\sf AOT\mhyphen HAP} = ({\sf CRSSetup}, {\sf DBSetup}, {\sf Issue}, {\sf Transfer})$
between a sender and multiple receivers.
\begin{description}
\item {\textbf{-- ${\sf CrsSetup}$:}} This randomized algorithm generates common reference string ${\sf crs}$ which is made public for everyone.
\item \textbf{-- ${\sf DBsetup}$:} This randomized algorithm is run by the sender $S$ who holds a database ${\sf DB} = \{m_i, W_i\}_{1 \leq i \leq N}$ of $N$ messages. Each message $m_i$ is associated with access policy $W_i = [W_{i, 1}, W_{i,2}, \ldots, W_{i,n}]$, where $W_{i, \ell} \subseteq S_{\ell}$, $i = 1, 2, \ldots, N, \ell = 1, 2, \ldots, n$. The algorithm generates public/secret key pair $({\sf pk}, {\sf sk})$, encrypts each $m_i$ under access policy $W_i$ and generates ciphertext $\phi_i$. The access policy $W_i$ is embedded implicitly in $\phi_i$. The algorithm outputs $({\sf pk}, {\sf sk}, \psi, {\sf cDB})$ to $S$. The proof $\psi$ guarantees that $({\sf pk}, {\sf sk})$ is a valid public/secret key pair. The sender $S$ publishes public key ${\sf pk}$, proof $\psi$ and cihertext database ${\sf cDB} = \{\phi_i\}_{1 \leq i \leq N}$ to all the receivers. The secret key ${\sf sk}$ and access policies $\{W_i\}_{1 \leq i \leq N}$ are kept hidden.
\item \textbf{-- ${\sf Issue}$ Protocol:} A receiver $R$ with attributes list $L = [L_1, L_2, \ldots, L_n]$, $L_{\ell} \in S_{\ell}$, engages in ${\sf Issue}$ protocol with $S$ to get attribute secret key ${\sf ASK}$ from $S$, where $\ell = 1, 2, \ldots, n$. At the end, $R$ gets ${\sf ASK}$ without leaking any information about its attribute list $L$ to $S$.
\item \textbf{-- ${\sf Transfer}$ Protocol:} The receiver $R$ invokes ${\sf transfer}$ protocol with $S$ to decrypt the ciphertext $\phi_{\sigma_j}$, where $\sigma_j \in \{1, 2, \ldots, N\}, j= 1, 2, \ldots, k$. The receiver $R$ recovers the message $m_{\sigma_j}$ if its attribute list $L$ satisfies the access policy $W_{\sigma_j}$ associated with $m_{\sigma_j}$, otherwise, $R$ gets some random message.
\end{description}

\subsection{UC Framework}\label{ucmodel} The \emph{universal composable (UC)} framework was presented by Canetti \cite{canetti2002universally} is adapted in this paper. It includes two worlds-- (i) \emph{real world} and (ii) \emph{ideal world}. The members involved in real world are $M$ receivers $R_1, R_2, \ldots, R_M$, a sender $S$, an adversary $\mathcal{A}$ and an \emph{environment machine} $\mathcal{Z}$. The environment machine $\mathcal{Z}$ starts the protocol execution $\Pi$ by giving inputs to all the parties. In real protocol execution $\Pi$, parties interact with each other and $\mathcal{A}$. Respective outputs are given $\mathcal{Z}$ by all the parties. Finally, $\mathcal{Z}$ outputs a bit, which is output of real world. Ideal world includes an additional incorruptible trusted party $\mathcal{F}_{\sf AOT\mhyphen HAP}^{N \times 1}$, called an ideal functionality. In ideal world, parties are dummy parties, $\widetilde{S}, \widetilde{R}_1, \widetilde{R}_2, \ldots, \widetilde{R}_M$. Dummy parties do not communicate with each other. Activated dummy party just forwards its input to $\mathcal{F}_{\sf AOT\mhyphen HAP}^{N \times 1}$ which is instructed to do the computation work. The ideal functionality $\mathcal{F}_{\sf AOT\mhyphen HAP}^{N \times 1}$ gives outputs to the dummy parties. At the end, $\mathcal{Z}$ outputs a bit, which is output of ideal world. 
Now, a protocol $\Pi$ is said to be securely realizes an ideal functionality $\mathcal{F}_{\sf AOT\mhyphen HAP}^{N \times 1}$ if $\mathcal{Z}$ is unable to distinguish its interaction with non-negligible probability between a real world adversary $\mathcal{A}$ and parties $S, R_1, R_2, \ldots, R_M$ running the protocol $\Pi$ in the real world $\mathcal{A}$ and an ideal world adversary $\mathcal{A'}$ and  $\mathcal{F}_{\sf AOT\mhyphen HAP}^{N \times 1}$, interacting with $\widetilde{S}, \widetilde{R}_1, \widetilde{R}_2, \ldots, \widetilde{R}_M$, in the ideal world.

\begin{definition}
%Let $\mathcal{F}_{\sf AOT\mhyphen HAP}^{N \times 1}$ be the adaptive oblivious transfer with hidden access policy
%functionality described above.
A protocol $\Pi$
securely realizes the ideal functionality $\mathcal{F}_{\sf AOT\mhyphen HAP}^{N \times 1} $ if for any $\mathcal{A}$, there
exists $\mathcal{A'}$ such that for any
$\mathcal{Z}$, $${ \sf IDEAL}_{\mathcal{F}_{\sf AOT\mhyphen HAP}^{N \times 1}, \mathcal{A'}, \mathcal{Z}} \stackrel{c}{\approx}
{\sf REAL}_{\Pi, \mathcal{A}, \mathcal{Z}},$$ where ${\sf
IDEAL}_{\mathcal{F}_{\sf AOT\mhyphen HAP}^{N \times 1}, \mathcal{A'},
\mathcal{Z}}$ and ${\sf
REAL}_{\Pi, \mathcal{A}, \mathcal{Z}}$ are ideal world and  real
world outputs respectively.
\end{definition}

\noindent \textbf{$\mathcal{F}_{{\sf CRS}}^{\mathcal{D}}$ Hybrid
Model:} As ideal functionalities can be UC-realized only in the presence of some trusted setup. One common form of setup is common reference string $({\sf CRS})$. Let us describe
the $\mathcal{F}_{{\sf CRS}}^{\mathcal{D}}$-hybrid model
\cite{canetti2002universally} that UC realizes a protocol
parameterized by some specific distribution
 $\mathcal{D}$. Upon receiving a message $({\sf CRS})$,
from any party, $\mathcal{F}_{{\sf CRS}}^{\mathcal{D}}$
first checks if there is a recorded value ${\sf crs}$. If not,
$\mathcal{F}_{{\sf CRS}}^{\mathcal{D}}$ generates ${\sf crs}
\leftarrow \mathcal{D}(1^{\rho})$ and records it. Finally,
$\mathcal{F}_{{\sf CRS}}^{\mathcal{D}}$ sends ${\sf
crs}$ to the party  and the adversary. In the proposed construction, the distribution $\mathcal{D}$ is ${\sf CrsSetup}$ algorithm, i.e, ${\sf crs}
\leftarrow {\sf CrsSetup}(1^{\rho})$. The algorithm outputs ${\sf crs}$ and is made public.
Let us briefly discuss how the parties interacts in the real world and in the ideal world. We discuss an interaction between the sender $S$ and a receiver $R$.

\noindent \textbf{Real World:} Upon receiving messages $({\sf initdb}, {\sf DB} = \{m_i$, $W_i\}_{1\leq i\leq N})$, $({\sf issue}, L)$ and $({\sf transfer}, \sigma_j)$ from $\mathcal{Z}$, how $S$ and $R$ behave are briefly discussed below, where $j = 1, 2, \ldots, k$.
\begin{itemize}
\item The sender $S$ upon receiving a message $({\sf initdb}, {\sf DB} = \{m_i, W_i\}_{1\leq i\leq N})$ runs ${\sf DBSetup}$ algorithm with input ${\sf crs}$ and ${\sf DB}$. In ${\sf DB} = \{m_i, W_i\}_{1\leq i\leq N}$, each $m_i$ is associated with access policy $W_i$. The algorithm outputs public/secret key pair $({\sf pk}, {\sf sk})$, a proof $\psi$ which shows the validity of public/secret key pair and ciphertext database ${\sf cDB} = \{\phi_i\}_{1\leq i\leq N}$. The access policy $W_i$ is embedded implicitly in $\phi_i$. The sender $S$ publishes ${\sf pk}, \psi, {\sf cDB}$ to all parties and keeps ${\sf sk}, \{W_i\}_{1\leq i\leq N}$ secret to itself.
\item The receiver $R$ upon receiving a message $({\sf issue}, L)$ from $\mathcal{Z}$ engages in ${\sf Issue}$ protocol with $S$. The receiver $R$ first randomizes its attribute set $L$ and generates randomized attributes set $L'$. The sender $S$ upon receiving $L'$ from $R$, generates randomized attribute secret key ${\sf ASK'}$ and gives ${\sf ASK'}$ to $R$. The receiver $R$ extracts ${\sf ASK}$ from ${\sf ASK'}$ using random values which were used to randomize the set $L$.
\item Upon receiving a message $({\sf transfer}, \sigma_j)$ from $\mathcal{Z}$, $R$ invokes ${\sf Transfer}$ protocol. The receiver $R$ takes ciphertext $\phi_{\sigma_j}$ from ${\sf cDB}$ and generates ${\sf Req}_{\sigma_j}, {\sf Pri}_{\sigma_j}, \pi_{\sigma_j}$. The request ${\sf Req}_{\sigma_j}$ is the randomization of $\phi_{\sigma_j}$ and is given to $S$. The proof $\pi_{\sigma_j}$ convinces $S$ that $R$ is decrypting a valid ciphertext $\phi_{\sigma_j} \in {\sf cDB}$. The receiver $R$ gives ${\sf Req}_{\sigma_j}, \pi_{\sigma_j}$ to $S$ and keeps ${\sf Pri}_{\sigma_j}$ secret to itself. The sender $S$ generates response ${\sf Res}_{\sigma_j}$ using secret key ${\sf sk}$. The response ${\sf Res}_{\sigma_j}$ and proof $\delta_{\sigma_j}$ is given to $R$ by $S$. The proof $\delta_{\sigma_j}$ convinces $R$ that $S$ has used the same secret key ${\sf sk}$ which was used in generating ${\sf cDB}$. Using ${\sf Res}_{\sigma_j}$, $R$ correctly decrypts $\phi_{\sigma_j}$ and recovers $m_{\sigma_j}$ if $R$'s attribute set $L$ implicitly satisfies $W_{\sigma_j}$. Otherwise, $R$ obtains a random message.
\end{itemize}
\noindent \textbf{Ideal World:} Upon receiving messages $({\sf initdb}, {\sf DB} = \{m_i$, $W_i\}_{1\leq i\leq N})$, $({\sf issue}, L)$ and $({\sf transfer}, \sigma_j)$ from $\mathcal{Z}$, the dummy parties forward them to the ideal functionality $\mathcal{F}_{\sf AOT\mhyphen HAP}^{N \times 1}$.
The ideal functionality $\mathcal{F}_{\sf AOT\mhyphen HAP}^{N \times 1}$ keeps an empty database ${\sf DB}$ and an empty attribute set $L_R$ for each receiver $R$. Let us briefly discuss the actions of $\mathcal{F}_{\sf AOT\mhyphen HAP}^{N \times 1}$.
\begin{itemize}
\item Upon receiving message $({\sf initdb}, {\sf DB} = \{m_i, W_i\}_{1\leq i\leq N})$ from $S$, $\mathcal{F}_{\sf AOT\mhyphen HAP}^{N \times 1}$ updates ${\sf DB} = \{m_i, W_i\}_{1\leq i\leq N}$.
\item The ideal functionality $\mathcal{F}_{\sf AOT\mhyphen HAP}^{N \times 1}$ upon receiving message $({\sf issue}, L)$ from $R$, sends $({\sf issue})$ to $S$ and receives a bit $b \in \{0, 1\}$ in return from $S$. If $b=1$, $\mathcal{F}_{\sf AOT\mhyphen HAP}^{N \times 1}$ updates $L_R = L$. Otherwise, $\mathcal{F}_{\sf AOT\mhyphen HAP}^{N \times 1}$ does nothing.
\item Upon receiving message $({\sf transfer}, \sigma_j)$ from $R$, $\mathcal{F}_{\sf AOT\mhyphen HAP}^{N \times 1}$ sends $({\sf transfer})$ to $S$, who in turn returns a bit $b \in \{0, 1\}$. If $b=1$, $\mathcal{F}_{\sf AOT\mhyphen HAP}^{N \times 1}$ checks whether $\sigma_j \in \{1, 2$, $\ldots, N\}$ and $L_R$ satisfies $W_{\sigma_j}$. If yes, $\mathcal{F}_{\sf AOT\mhyphen HAP}^{N \times 1}$ gives $m_{\sigma_j}$ to $R$. Otherwise, $\mathcal{F}_{\sf AOT\mhyphen HAP}^{N \times 1}$ sends a null string $\bot$ to $R$.
\end{itemize}

\section{Protocol}\label{protocoldescription}
Our ${\sf AOT\mhyphen HAP} = ({\sf CRSSetup}, {\sf DBSetup}, {\sf Issue}, {\sf Transfer})$ couples Boneh-Boyen (BB) signature \cite{boneh2004shortsignature} and ciphertext policy attribute based encryption (CP-ABE) of Nishide {\em et al.} \cite{nishide2008attribute}. In addition, Groth-Sahai \cite{groth2008efficient} proofs for non-interactive verification of pairing product equations are employed. The ${\sf AOT\mhyphen HAP}$ involves a sender and multiple receivers. We consider an interaction between a sender $S$ with a database ${\sf DB} = \{m_i, W_i\}_{1 \leq i\leq N}$ and a receiver $R$ with attribute list $L = [v_{1, t_1}, v_{2, t_2}, \ldots, v_{n, t_n}]$. As a protocol can be universal composable (UC) realized only under some trusted setup assumptions, algorithm ${\sf CRSSetup}$ is invoked to generate a common reference string ${\sf crs}$ which is made public to all parties. In initialization phase, S executes algorithm ${\sf DBSetup}$. The sender $S$ holds a database ${\sf DB} = \{m_i, W_i\}_{1 \leq i \leq N}$, where $W_i$ is the access policy associated with $m_i$, $1 \leq i \leq N$. Each message $m_i$ is encrypted to generate ciphertext $\phi_i$ using CP-ABE of \cite{nishide2008attribute} under $W_i$. The access policy $W_i$ is embedded implicitly in $\phi_i$ and is not made public. The BB signature is used to sign the random value which was used to encrypt $m_i$. In Issue phase, the receiver $R$ with the attribute list $L$ interacts with $S$ to get attribute secret key ${\sf ASK}$. To decrypt the ciphertext $\phi_{\sigma_j}$, $R$ engages in ${\sf Transfer}$ protocol with $S$ and extracts $m_{\sigma_j}$ if $L \models W_{\sigma_j}$, otherwise, $R$ gets a random message. Formal description of ${\sf CRSSetup}, {\sf DBSetup}, {\sf Issue}, {\sf Transfer}$ are given below. Let $\mathbb{A}_1, \mathbb{A}_2, \ldots, \mathbb{A}_n$ be $n$ attributes and each attribute $\mathbb{A}_{\ell}$ can take $n_{\ell}$ values, namely, $\{v_{\ell, 1}, v_{\ell, 2}, \ldots, v_{\ell, n_{\ell}}\}$ as explained in section \ref{access structure}, $\ell = 1, 2, \ldots, n$.

\begin{algorithm}[H]
\caption {~~${\sf CrsSetup}$}\label{chap6:algo1}
%\scriptsize
\begin{algorithmic}[1]
    \Require {$\rho, n, n_1, n_2, \ldots, n_n$.}
    \Ensure {${\sf crs}$.}
    \State {Generate ${\sf Assparams} \leftarrow {\sf AssBilinearSetup}(1^{\rho})$};
    \For {$(\ell =0, 1, \ldots, n)$}
    \For {$(t =0, 1, \ldots, n_{\ell})$}
    \State {Choose $a_{\ell, t} \xleftarrow{\$} \mathbb{Z}_{p}$};
    \State {Set $A_{\ell, t} = g_{1}^{a_{\ell, t}}, B_{\ell, t} = g_{2}^{a_{\ell, t}}$};
    \EndFor
    \EndFor
    \State {Take $\xi_1, \xi_2, \widetilde{a}, \widetilde{b}, \widehat{\xi}_1, \widehat{\xi}_2, \widehat{a}, \widehat{b} \xleftarrow{\$} \mathbb{Z}_{p}$};
    \State {Set $u_1 = (g_{1}^{\widetilde{a}}, 1, g_1), u_2 = (1, g_{1}^{\widetilde{b}}, g_1)$, $u_3 = u_{1}^{\xi_1}u_{2}^{\xi_2} = (g_{1}^{\widetilde{a}\xi_1}, g_{1}^{\widetilde{b}\xi_2}, g_{1}^{\xi_1+\xi_2})$};
    \State {Set $v_1 = (g_{2}^{\widetilde{a}}, 1, g_{2}), v_2 = (1, g_{2}^{\widetilde{b}}, g_{2}),$ $v_3 = v_{1}^{\xi_1}v_{2}^{\xi_2} = (g_{2}^{\widetilde{a}\xi_1}, g_{2}^{\widetilde{b}\xi_2}, g_{2}^{\xi_1+\xi_2})$};
    \State {$\widehat{u}_1 = (g_{1}^{\widehat{a}}, 1, g_1), \widehat{u}_2 = (1, g_{1}^{\widehat{b}}, g_1),$ $\widehat{u}_3 = \widehat{u}_{1}^{\widehat{\xi}_1}\widehat{u}_{2}^{\widehat{\xi}_2} = (g_{1}^{\widehat{a}\widehat{\xi}_1}, g_{1}^{\widehat{b}\widehat{\xi}_2}, g_{1}^{\widehat{\xi}_1+\widehat{\xi}_2})$};
    \State {$\widehat{v}_1 = (g_{2}^{\widehat{a}}, 1, g_2), \widehat{v}_2 = (1, g_{2}^{\widehat{b}}, g_2),$
$\widehat{v}_3 = \widehat{v}_{1}^{\widehat{\xi}_1}\widehat{v}_{2}^{\widehat{\xi}_2} = (g_{2}^{\widehat{a}\widehat{\xi}_1}, g_{2}^{\widehat{b}\widehat{\xi}_2}, g_{2}^{\widehat{\xi}_1+\widehat{\xi}_2})$};
 \State {${\sf GS}_S = (u_1, u_2, u_3, v_1, v_2, v_3), {\sf GS}_R = (\widehat{u}_1, \widehat{u}_2, \widehat{u}_3, \widehat{v}_1, \widehat{v}_2, \widehat{v}_3)$};
  \State {${\sf crs} = ({\sf Assparams}, {\sf GS}_S, {\sf GS}_R, \{A_{\ell, t}, B_{\ell, t}\}_{1 \leq t \leq n_{\ell}, 1 \leq \ell \leq n})$};
\end{algorithmic}
\end{algorithm}

\noindent \textbf{-- ${\sf CrsSetup}(1^{\rho})$.} The algorithm generates common reference string ${\sf crs}$ on input security parameter $\rho$, number of attributes $n$, number of values $n_{\ell}$  that each attribute $\mathbb{A}_{\ell}$ can take. The algorithm generates bilinear pairing groups by invoking algorithm ${\sf BilinearSetup}$ given in section \ref{preliminaries}. For each attribute value $v_{\ell, t}$, the algorithm picks $a_{\ell, t}$ randomly from $\mathbb{Z}_{p}$, sets $A_{\ell, t} = g_{1}^{a_{\ell, t}}, B_{\ell, t} = g_{2}^{a_{\ell, t}}$, for $t = 1, 2, \ldots, n_{\ell}$, $\ell =1, 2, \ldots, n$. The algorithm generates the Groth-Sahai common reference strings ${\sf GS}_S = (u_1, u_2, u_3, v_1, v_2, v_3)$ for $S$ and ${\sf GS}_R = (\widehat{u}_1, \widehat{u}_2, \widehat{u}_3, \widehat{v}_1, \widehat{v}_2, \widehat{v}_3)$ for $R$ in perfectly sound setting by choosing $\xi_1, \xi_2, \widetilde{a}, \widetilde{b}, \widehat{\xi}_1, \widehat{\xi}_2, \widehat{a}, \widehat{b}$ randomly from $\mathbb{Z}_{p}$ and setting $u_1 = (g_{1}^{\widetilde{a}}, 1, g_1), u_2 = (1, g_{1}^{\widetilde{b}}, g_1), u_3 = (g_{1}^{\widetilde{a}\xi_1}, g_{1}^{\widetilde{b}\xi_2}, g_{1}^{\xi_1+\xi_2})$,
$v_1 = (g_{2}^{\widetilde{a}}, 1, g_{2}), v_2 = (1, g_{2}^{\widetilde{b}}, g_{2}), v_3 = (g_{2}^{\widetilde{a}\xi_1}, g_{2}^{\widetilde{b}\xi_2}, g_{2}^{\xi_1+\xi_2})$,
$\widehat{u}_1 = (g_{1}^{\widehat{a}}, 1, g_1), \widehat{u}_2 = (1, g_{1}^{\widehat{b}}, g_1), \widehat{u}_3 = (g_{1}^{\widehat{a}\widehat{\xi}_1}, g_{1}^{\widehat{b}\widehat{\xi}_2}, g_{1}^{\widehat{\xi}_1+\widehat{\xi}_2})$, $\widehat{v}_1 = (g_{2}^{\widehat{a}}, 1, g_2), \widehat{v}_2 = (1, g_{2}^{\widehat{b}}, g_2), \widehat{v}_3 = (g_{2}^{\widehat{a}\widehat{\xi}_1}, g_{2}^{\widehat{b}\widehat{\xi}_2}, g_{2}^{\widehat{\xi}_1+\widehat{\xi}_2})$. It outputs ${\sf crs} = ({\sf params}, {\sf GS}_S, {\sf GS}_R$, $\{\{A_{\ell, t}, B_{\ell, t}\}_{1 \leq t \leq n_{\ell}}\}_{1 \leq \ell \leq n})$ to all parties.

\begin{algorithm}[H]
\caption {~~${\sf DBSetup}$}\label{chap6:algo2}
%\scriptsize
\begin{algorithmic}[1]
    \Require {${\sf crs}, {\sf DB} = \{(m_i, W_i)\}_{1\leq i \leq N}$.}
    \Ensure {${\sf pk}, {\sf sk}, \psi, {\sf cDB} = \{\phi_i\}_{1\leq i \leq N}$.}
    \State  {Randomly choose $w, x, \alpha, \beta, \gamma \xleftarrow{\$} \mathbb{Z}_{p}$};
    \State {Compute $B = g_{1}^{\beta}, y = g_{1}^{x}, h = g_{2}^{\gamma}, Y = e(g_1, g_{2}^{w}),$ $H = e(B, h), B^{\gamma}, g_{1}^{w}, g_{1}^{\alpha}, g_{2}^{\alpha}$};
    \State {Generate ${\sf Com'}(h)$ using ${\sf GS}_R$};
    \State {Set public key ${\sf pk} = (B, y, Y, H$, ${\sf Com'}(h))$};
    \State {Set secret key ${\sf sk} = (x,\alpha, \beta, \gamma, w, h, B^{\gamma}, g_{1}^{w}, g_{2}^{w}, g_{1}^{\alpha}, g_{2}^{\alpha})$};
    \State {The proof $\psi = {\sf NIZK}_{{\sf GS}_S}\{(B^{\gamma}, g_{1}^{w}, g_{2}^{w}, h, g_{1}^{\alpha}, g_{2}^{\alpha}, g')~|~ $
$e(g_{1}^{-1}, g_{2}^{w})e(g_{1}^{w}, g') =1
\wedge e(B^{-1}, h)e(B^{\gamma}, g') = 1 \wedge e(g_{1}^{-1}, g_{2}^{\alpha})e(g_{1}^{\alpha}, g')=1 \wedge e(g_1, g') = e(g_1, g_2)\}$};
\For{$(i =1, 2, \ldots, N)$}
\State {$s_{i, 1}, s_{i,2}, \ldots, s_{i,n} \xleftarrow{\$} \mathbb{Z}_{p}$};
\State {Set $r_i = s_{i, 1} + s_{i,2} + \ldots + s_{i,n} $};
\State {$c_{i}^{(1)} = g_{2}^{\frac{1}{x+r_i}}$};
\State {$c_{i}^{(2)} = B^{r_i}$};
\State {$c_{i}^{(3)} = m_i\cdot Y^{r_i}\cdot e(c_{i}^{(2)}, h)$};
\For {$(\ell =1, 2, \ldots, n)$}
\State {$c_{i, \ell}^{(4)} = g_{1}^{s_{i, \ell}}$}
\For {$(t =1, 2, \ldots, n_{\ell})$}
\If {$(v_{\ell, t} \in W_{i, \ell})$}
\State {$c_{i, \ell, t}^{(5)} = A_{\ell, t}^{\alpha s_{i, \ell}}$, where $A_{\ell, t}$ is extracted from ${\sf crs}$};
\Else
\State {$c_{i, \ell, t}^{(5)} = g_{1}^{z_{i, \ell, t}}, z_{i, \ell, t}\xleftarrow{\$} \mathbb{Z}_{p}$};
  \EndIf
\EndFor
\EndFor
\State {Set $c_{i}^{(4)} = \{c_{i, \ell}^{(4)}\}_{1\leq \ell \leq n}, c_{i}^{(5)} = \{\{c_{i, \ell, t}^{(5)}\}_{1\leq t\leq n_{\ell}}\}_{1\leq \ell \leq n}$};
\State {$\phi_i = (c_{i}^{(1)}, c_{i}^{(2)}, c_{i}^{(3)}, c_{i}^{(4)}, c_{i}^{(5)})$};
\EndFor
\State {${\sf cDB} = (\phi_1, \phi_2, \ldots, \phi_N)$};
\end{algorithmic}
\end{algorithm}

\noindent \textbf{-- ${\sf DBSetup}({\sf crs}, {\sf DB} = \{(m_i, W_i)\}_{1\leq i \leq N})$.} This algorithm is executed by $S$ with input ${\sf crs}$, database ${\sf DB} = \{(m_i$, $W_i)\}_{1\leq i \leq N}$, where $m_i \in \mathbb{G}_T$, $W_i = [W_{i,1}, W_{i,2}, \ldots, W_{i,n}]$ is the access policy associated with $m_i$, $W_{i, \ell} \subseteq S_{\ell}$ as discussed in section \ref{access structure}, $i =1, 2, \ldots, N$, $\ell =1, 2, \ldots, n$. The algorithm randomly takes $w, x, \alpha, \beta, \gamma$ from $\mathbb{Z}_{p}$, computes $B = g_{1}^{\beta}, y = g_{1}^{x}, h = g_{2}^{\gamma}, Y = e(g_1, g_{2}^{w}), H = e(B, h), B^{\gamma}, g_{1}^{w}$, $g_{1}^{\alpha}, g_{2}^{\alpha}$. Additionally, it computes ${\sf Com'}(h)$ using ${\sf GS}_R$ as in section \ref{noninteractive}. The algorithm sets public key ${\sf pk} = (B, y, Y, H$, ${\sf Com'}(h))$, secret key ${\sf sk} = (x,\alpha, \beta, \gamma, w, h, B^{\gamma}, g_{1}^{w}, g_{2}^{w}, g_{1}^{\alpha}, g_{2}^{\alpha})$. The proof
\begin{gather*}
\psi = {\sf NIZK}_{{\sf GS}_S} \{(B^{\gamma}, g_{1}^{w}, g_{2}^{w}, h, g_{1}^{\alpha}, g_{2}^{\alpha}, g')~|~ e(g_{1}^{-1}, g_{2}^{w})e(g_{1}^{w}, g') =1\wedge e(B^{-1}, h)e(B^{\gamma}, g') = 1 \wedge \\
(g_{1}^{-1}, g_{2}^{\alpha})e(g_{1}^{\alpha}, g')=1 \wedge e(g_1, g') = e(g_1, g_2)\}
\end{gather*}
is generated to convince $R$ that $S$ knows secrets $g_{2}^{w}, g_{2}^{\alpha}, h$. The proof $\psi$ consists of proof components of equations $e(g_{1}^{-1}, g_{2}^{w})e(g_{1}^{w}, g') =1\wedge e(B^{-1}, h)e(B^{\gamma}$, $g')$ $= 1 \wedge (g_{1}^{-1}, g_{2}^{\alpha})e(g_{1}^{\alpha}, g')=1\wedge e(g_1, g') = e(g_1, g_2)$ and commitments of witnesses $B^{\gamma}, g_{1}^{w}, g_{2}^{w}, h, g_{1}^{\alpha}, g_{2}^{\alpha}, g'$ generated using ${\sf GS}_S$ as explained in section \ref{noninteractive}. To encrypt $m_i$ under $W_i$, random values $s_{i, 1}, s_{i,2}, \ldots, s_{i,n}$ are selected from $ \mathbb{Z}_{p}$. The algorithm sets $r_i = s_{i, 1} + s_{i,2} + \ldots + s_{i,n}$ and signs $r_i$ using Boneh-Boyen (BB) signature \cite{boneh2004shortsignature} as shown in line 10 of Algorithm 2. After that, $m_i$ is encrypted using CP-ABE of \cite{nishide2008attribute}.

\noindent In ciphertext $\phi_i = (c_{i}^{(1)}, c_{i}^{(2)}, c_{i}^{(3)}, c_{i}^{(4)}, c_{i}^{(5)})$, $c_{i}^{(1)} = g_{2}^{\frac{1}{x+r_i}}$ is the signature on $r_i$, $c_{i}^{(2)}, c_{i}^{(3)}, c_{i}^{(4)}, c_{i}^{(5)}$ is the CP-ABE of $m_i$, $i =1, 2, \ldots, N$. The algorithm outputs ${\sf pk}, {\sf sk}, \psi$ and ciphertext database ${\sf cDB} = \{\phi_i\}_{1\leq i \leq N}$ to $S$. The sender $S$ publishes ${\sf pk}, \psi, {\sf cDB} = \{\phi_i\}_{1\leq i \leq N}$ to all parties and keeps ${\sf sk}$ hidden.

\noindent \textbf{Correctness of ${\sf cDB}$:} Anyone can verify the correctness of ${\sf cDB} = \{\phi_i\}_{1\leq i \leq N}$, $\phi_i = (c_{i}^{(1)}, c_{i}^{(2)}, c_{i}^{(3)}$, $c_{i}^{(4)} = \{c_{i, \ell}^{(4)}\}_{1\leq \ell \leq n}, c_{i}^{(5)})$ by verifying the following equations
$$e\left(\prod_{\ell=1}^{n}c_{i, \ell}^{(4)}\cdot y,  c_{i}^{(1)}\right) = e(g_1, g_2) ~~\forall~ i = 1, 2, \ldots, N.$$

\begin{figure}
\centering
\epsfig{file=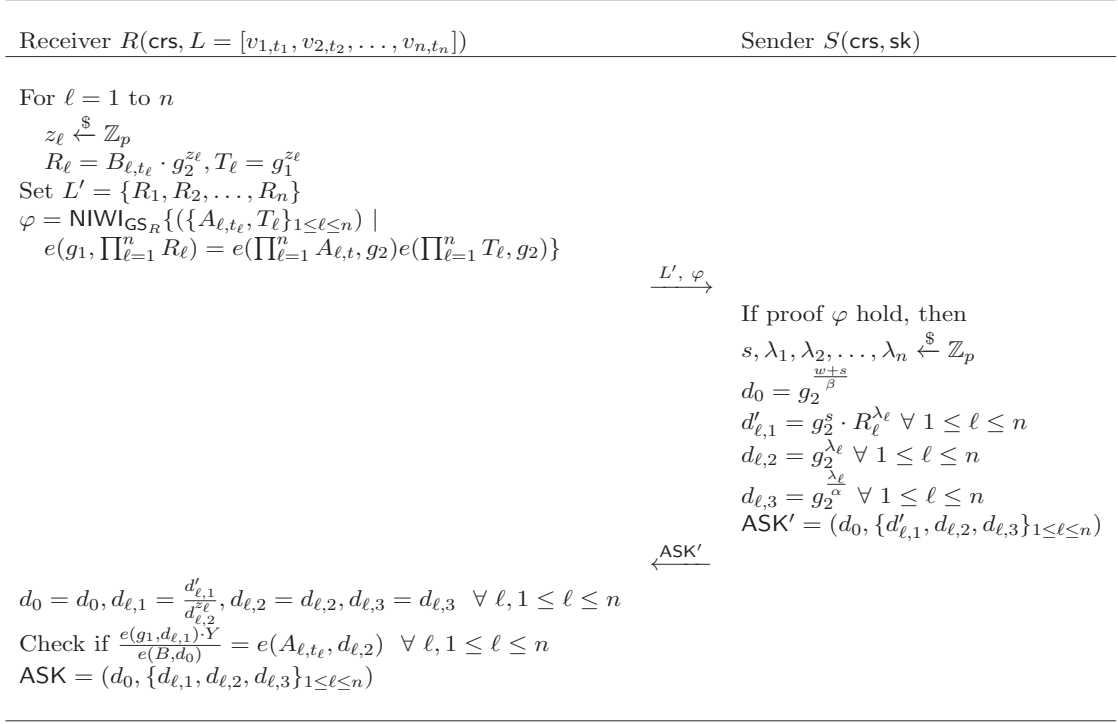, height=4in, width=6in,}
\caption{Issue Phase.} \label{fig1}
\end{figure}

\noindent \textbf{-- ${\sf Issue}$ Protocol$({\sf crs}, L, {\sf sk})$.} The receiver $R$ invokes ${\sf Issue}$ protocol, given in Figure \ref{fig1}, to obtain attribute secret key ${\sf ASK}$ for its attributes list $L = [v_{1, t_1}, v_{2, t_2}, \ldots, v_{n, t_n}]$, where $v_{\ell, t_{\ell}} \in S_{\ell}, \ell =1, 2, \ldots, n$ as given in section \ref{access structure}. To achieve this, $R$ randomizes $B_{\ell, t_{\ell}}$, extracted from ${\sf crs} = ({\sf params}, {\sf GS}_S$, ${\sf GS}_R$, $\{\{A_{\ell, t}, B_{\ell, t}\}_{1 \leq t \leq n_{\ell}}\}_{1 \leq \ell \leq n})$, corresponding to each $v_{\ell, t_{\ell}}$ from $L$ using random value $z_{\ell}$ and generates randomized set $L'$ as shown in Figure \ref{fig1}, where $t_{\ell} =1, 2, \ldots, n_{\ell}$, $\ell = 1, 2, \ldots, n$.
The receiver $R$ sets proof $\varphi = {\sf NIWI}_{{\sf GS}_R}\{(\{A_{\ell, t_{\ell}}$, $T_{\ell}\}_{1 \leq \ell \leq n})~|~ e(g_1, \prod_{\ell=1}^{n}R_{\ell}) = e(\prod_{\ell=1}^{n}A_{\ell,t}$, $g_2)e(\prod_{\ell=1}^{n}T_{\ell}, g_2)\}$ to convince $R$ that $L'$ is the randomization of $L$. The randomized set $L'$ and proof $\varphi$ is given to $S$. If the proof $\varphi$ is valid, $S$ uses secret $w, \alpha, \beta$ from its secret key ${\sf sk} = (x,\alpha, \beta, \gamma, w, h, B^{\gamma}, g_{1}^{w}, g_{2}^{w}, g_{1}^{\alpha}, g_{2}^{\alpha})$ to generate $d_0 = g_{2}^{\frac{w+s}{\beta}}$, $d'_{\ell, 1} = g_{2}^{s}\cdot R_{\ell}^{\lambda_{\ell}}$, $d_{\ell, 2} = g_{2}^{\lambda_{\ell}}$, $d_{\ell, 3} = g_{2}^{\frac{\lambda_{\ell}}{\alpha}}$, where $s, \lambda_{\ell}$ are random values from $\mathbb{Z}_{p}$, $\ell =1, 2, \ldots, n$. The randomized attribute secret key ${\sf ASK'} = (d_0, \{d'_{\ell, 1}, d_{\ell, 2}, d_{\ell, 3}\}_{1\leq \ell \leq n})$ is given to $R$. The receiver $R$ extracts ${\sf ASK}$ from ${\sf ASK'}$ using random values $z_1, z_2, \ldots, z_n$ which were used by $R$ to randomize $L$ and checks the correctness of ${\sf ASK}$ by verifying the equation
$$\frac{e(g_1, d_{\ell,1})\cdot Y}{e(B, d_0)} = e(A_{\ell, t_{\ell}}, d_{\ell,2}) ~~~~~\forall~\ell,  1\leq \ell \leq n.$$
%\begin{figure}[htbp]
%%\scriptsize
%\centering
%\begin{tabular}{l c l}
%\cline{1-3}
%&&\\
%Receiver $R ({\sf crs}, L = [v_{1, t_1}, v_{2, t_2}, \ldots, v_{n, t_n}])$ & & Sender $S({\sf crs}, {\sf sk})$  \\
%\cline{1-3}
%&&\\
%For $\ell =1$ to $n$&&\\
%~~~$z_{\ell} \xleftarrow{\$} \mathbb{Z}_{p}$&&\\
%~~~$R_{\ell} = B_{\ell, t_{\ell}}\cdot g_{2}^{z_{\ell}}, T_{\ell} = g_{1}^{z_{\ell}}$&&\\
%Set $L' = \{R_1, R_2, \ldots, R_n\}$&&\\
%$\varphi = {\sf NIWI}_{{\sf GS}_R}\{(\{A_{\ell, t_{\ell}}, T_{\ell}\}_{1\leq \ell \leq n})~|~ $&&\\
%$~~~e(g_1, \prod_{\ell=1}^{n}R_{\ell}) = e(\prod_{\ell=1}^{n}A_{\ell,t}, g_2)e(\prod_{\ell=1}^{n}T_{\ell}, g_2)\}$&&\\
%&$\xlongrightarrow {L', ~\varphi}$ &\\
%&&If proof $\varphi$ hold, then\\
%&&$s, \lambda_1, \lambda_2, \ldots, \lambda_n \xleftarrow{\$} \mathbb{Z}_{p}$\\
%&&$d_0 = g_{2}^{\frac{w+s}{\beta}}$\\
%&&$d'_{\ell, 1} = g_{2}^{s}\cdot R_{\ell}^{\lambda_{\ell}} ~\forall~ 1\leq \ell \leq n $\\
%&&$ d_{\ell, 2} = g_{2}^{\lambda_{\ell}} ~\forall~ 1\leq \ell \leq n$\\
%&&${\sf ASK'} = (d_0, \{d'_{\ell, 1}, d_{\ell, 2}\}_{1\leq \ell \leq n})$\\
%&$\xlongleftarrow {{\sf ASK'}}$&\\
%$d_0 = d_0, d_{\ell, 1} = \frac{d'_{\ell, 1}}{d_{\ell, 2}^{z_{\ell}}}, d_{\ell, 2} = d_{\ell, 2}~~\forall~\ell,  1\leq \ell \leq n$&&\\\\
%\hline
%\end{tabular}
%\vspace{2mm} \caption{Issue Phase.} \label{fig1}
%\vskip -6pt
%\end{figure}

\begin{figure*}
\centering
\epsfig{file=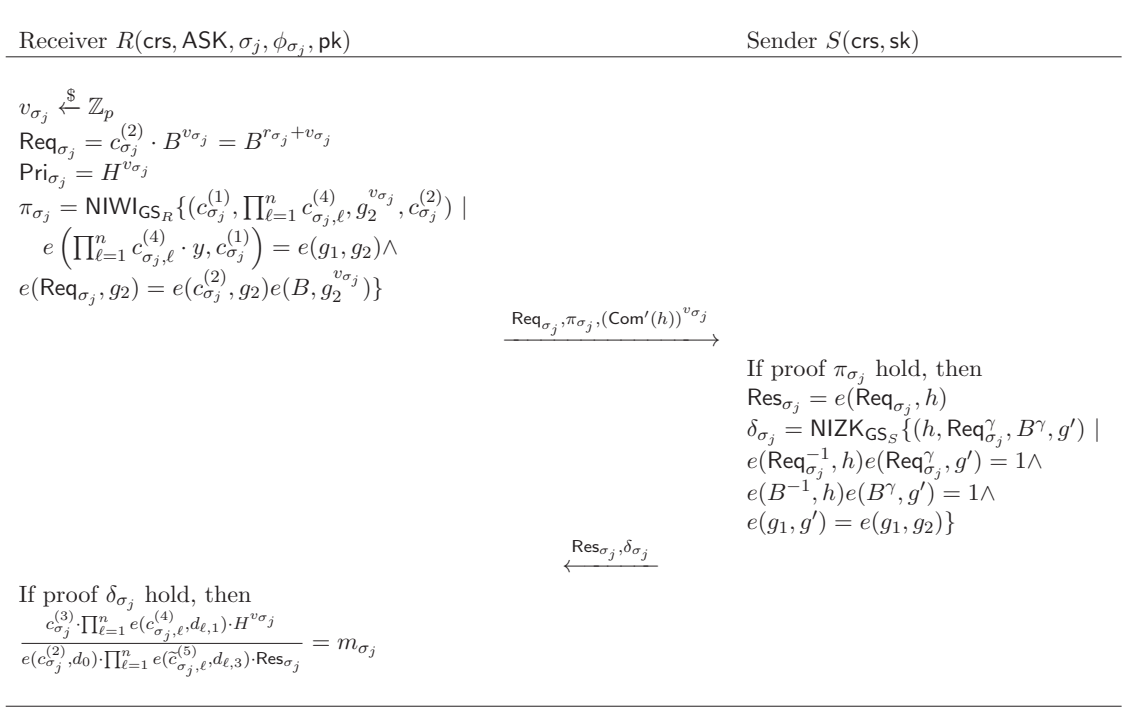, height=4in, width=6in,}
\caption{Transfer Phase.} \label{fig2}
\end{figure*}

\noindent \textbf{-- ${\sf Transfer}$ Protocol$({\sf crs}, {\sf ASK}, {\sf pk}, {\sf sk}, \sigma_j, \phi_{\sigma_j})$.} The receiver $R$ engages in ${\sf Transfer}$ protocol with $S$ to decrypt the ciphertext $\phi_{\sigma_j} = (c_{\sigma_j}^{(1)}, c_{\sigma_j}^{(2)}, c_{\sigma_j}^{(3)}, c_{\sigma_j}^{(4)}, c_{\sigma_j}^{(5)})$ as shown in Figure \ref{fig2}. For this, $R$ randomizes $c_{\sigma_j}^{(2)} = B^{r_{\sigma_j}}$ using random value $v_{\sigma_j}$ from $\mathbb{Z}_{p}$ and generates request ${\sf Req}_{\sigma_j} = c_{\sigma_j}^{(2)}\cdot B^{v_{\sigma_j}} =
B^{r_{\sigma_j} + v_{\sigma_j}}$, private ${\sf Pri}_{\sigma_j} = H^{v_{\sigma_j}}$ and sets proof
$\pi_{\sigma_j}$ $= {\sf NIWI}_{{\sf GS}_R}\{(c_{\sigma_j}^{(1)}, \prod_{\ell=1}^{n}c_{\sigma_j, \ell}^{(4)}, g_{2}^{v_{\sigma_j}}, c_{\sigma_j}^{(2)})~|~$
$e\left(\prod_{\ell=1}^{n}c_{\sigma_j, \ell}^{(4)}\cdot y, c_{\sigma_j}^{(1)}\right)$ $= e(g_1, g_2) \wedge
e({\sf Req}_{\sigma_j}, g_2) = e(c_{\sigma_j}^{(2)}, g_2)e(B, g_{2}^{v_{\sigma_j}})\}$. The first equation $e(\prod_{\ell=1}^{n}c_{\sigma_j, \ell}^{(4)}\cdot y, c_{\sigma_j}^{(1)})$ $= e(g_1, g_2)$ in $\pi_{\sigma_j}$
guarantees $S$ that $\phi_{\sigma_j}$ is a valid ciphertext, i.e, $\phi_{\sigma_j} \in {\sf cDB}$, and 2nd equation shows that ${\sf Req}_{\sigma_j}$ is the randomization of $c_{\sigma_j}^{(2)}$. The receiver also raises power $v_{\sigma_j}$ to ${\sf Com'}(h)$ which was given by $S$ in initialization phase generated using ${\sf GS}_R$. This extra component does not give any advantage to adversary and is useful only for simulator during security analysis. The sender $S$ checks the validity of $\pi_{\sigma_j}$. If invalid, $S$ outputs a null string $\bot$. Otherwise, $S$
uses secret $h, \gamma$ from secret key ${\sf sk} = (x,\alpha, \beta, \gamma, w, h, B^{\gamma}, g_{1}^{w}, g_{2}^{w}, g_{1}^{\alpha}, g_{2}^{\alpha})$ to compute response
${\sf Res}_{\sigma_j} = e({\sf Req}_{\sigma_j}, h)$. The proof $\delta_{\sigma_j} = {\sf NIZK}_{{\sf GS}_S}\{(h, {\sf Req}_{\sigma_j}^{\gamma}, B^{\gamma}, g')~|~
e({\sf Req}_{\sigma_j}^{-1}, h)e({\sf Req}_{\sigma_j}^{\gamma}$, $g') = 1 \wedge
e(B^{-1}$, $h)e(B^{\gamma}, g') = 1 \wedge
e(g_1, g') = e(g_1, g_2)\}$
is generated to guarantee $R$ that $S$ has used the same secrets $h, \gamma$ which were used in initialization phase to set ciphertext database ${\sf cDB}$. Upon receiving ${\sf Res}_{\sigma_j}, \delta_{\sigma_j}$, $R$ first verifies $\delta_{\sigma_j}$. If fails, $R$ outputs a null string $\bot$. Otherwise, $R$ sets $\widetilde{c}_{\sigma_j, \ell}^{(5)} = c_{\sigma_j, \ell, t_{\ell}}^{(5)}$ if $v_{\ell, t_{\ell}} \in L$ for $\ell =1, 2, \ldots, n$ and computes
\vspace{-.3cm}
\begin{equation}\label{eq4}
\frac{c_{\sigma_j}^{(3)}\cdot \prod_{\ell=1}^{n} e(c_{\sigma_j, \ell}^{(4)}, d_{\ell, 1})\cdot H^{v_{\sigma_j}}}{e(c_{\sigma_j}^{(2)}, d_0)\cdot \prod_{\ell=1}^{n} e(\widetilde{c}_{\sigma_j, \ell}^{(5)}, d_{\ell, 3})\cdot{\sf Res}_{\sigma_j}} = m_{\sigma_j}.
\end{equation}

\noindent \textbf{Correctness of equation \ref{eq4}.}
\begin{align*}
&{\sf val}_1 = \prod_{\ell=1}^{n} e(c_{\sigma_j, \ell}^{(4)}, d_{\ell, 1}) = \prod_{\ell=1}^{n}e(g_{1}^{s_{\sigma_j, \ell}}, g_{2}^{s}\cdot B_{\ell, t_{\ell}}^{\lambda_{\ell}})\\
& = \prod_{\ell=1}^{n}e(g_{1}^{s_{\sigma_j, \ell}},g_{2}^{s + a_{\ell, t_{\ell}}\lambda_{\ell}})\\
&= e(g_1, g_2)^{s\sum_{\ell=1}^{n} s_{\sigma_j, \ell}}e(g_1, g_2)^{\sum_{\ell=1}^{n} s_{\sigma_j, \ell}a_{\ell, t_{\ell}}\lambda_{\ell}}
\end{align*}
\begin{align*}
&= e(g_1, g_2)^{sr_{\sigma_j}}e(g_1, g_2)^{\sum_{\ell=1}^{n} s_{\sigma_j, \ell}a_{\ell, t_{\ell}}\lambda_{\ell}} {\rm ~as~} \sum_{\ell=1}^{n} s_{\sigma_j, \ell} = r_{\sigma_j}
\\
&{\sf val}_2 = \prod_{\ell=1}^{n} e(\widetilde{c}_{\sigma_j, \ell}^{(5)}, d_{\ell, 3})
\\
&= \prod_{\ell=1}^{n} e(A_{\ell, t_{\ell}}^{\alpha s_{\sigma_j, \ell}}, g_{2}^{\frac{\lambda_{\ell}}{\alpha}}) = \prod_{\ell=1}^{n}e(g_{1}^{a_{\ell, t_{\ell}}s_{\sigma_j, \ell}}, g_{2}^{\lambda_{\ell}})\\
&=e(g_1, g_2)^{\sum_{\ell=1}^{n} s_{\sigma_j, \ell}a_{\ell, t_{\ell}}\lambda_{\ell}}\\
&\frac{{\sf val}_1}{e(c_{\sigma_j}^{(2)}, d_0){\sf val}_2} = \frac{e(g_1, g_2)^{sr_{\sigma_j}}}{e(g_1, g_2)^{r_{\sigma_j}(w+s)}} = \frac{1}{e(g_1, g_2)^{r_{\sigma_j}w}} = \frac{1}{Y^{r_{\sigma_j}}}\\
&\frac{c_{\sigma_j}^{(3)}\cdot {\sf val}_1\cdot H^{v_{\sigma_j}}}{e(c_{\sigma_j}^{(2)}, d_0)\cdot {\sf val}_2\cdot{\sf Res}_{\sigma_j}}
= \frac{m_{\sigma_j} Y^{r_{\sigma_j}}e(B,h)^{r_{\sigma_j}}H^{v_{\sigma_j}}}{Y^{r_{\sigma_j}}e(B, h)^{r_{\sigma_j + v_{\sigma_j}}}} = m_{\sigma_j},
%&= \frac{m_{\sigma_j}\cdot Y^{r_{\sigma_j}}\cdot e(g_1, g_2)^{sr_{\sigma_j}}}{Y^{r_{\sigma_j}}\cdot e(g_1, g_2)^{sr_{\sigma_j}}}
% = m_{\sigma_j},
\end{align*}
where $c_{\sigma_j}^{(3)} = m_{\sigma_j}\cdot Y^{r_{\sigma_j}}\cdot e(B, h)^{r_{\sigma_j}},H = e(B, h)$, ${\sf Res}_{\sigma_j} = e(B, h)^{r_{\sigma_j + v_{\sigma_j}}}$.

\begin{theorem}\label{theorem1}
The adaptive oblivious transfer with hidden access policy (${\sf AOT\mhyphen HAP}$) presented in section
\ref{protocoldescription} securely realizes the ideal
functionality $\mathcal{F}_{\sf AOT\mhyphen HAP}^{N \times 1}$ in the $\mathcal{F}_{\sf
CRS}^{\mathcal{D}}$-hybrid model described in section \ref{securitymodel}
under DLIN, DBDH and $q$-SDH assumptions.
\end{theorem}
\begin{proof}
The ${\sf AOT\mhyphen HAP}$ is run between a sender $S$ and multiple receivers. We simulate an interaction between $S$ and a receiver $R$ with an attribute list $L = [L_1, L_2, \ldots, L_n]$, $L_{\ell} \in S_{\ell}$, $n$ is the total number of attributes, $\ell=1, 2, \ldots, n$, (see section \ref{access structure} for detail definition of $L_{\ell}, S_{\ell}$). The adversarial model adapted in this paper is static corruption model in which adversary decides which party to corrupt before the protocol starts. The corrupted parties remain corrupt, and honest parties remain honest throughout. Therefore, the case of malicious sender and malicious receiver are addressed separately. In both cases, an ideal world adversary $\mathcal{A'}$ corresponding to a real world adversary $\mathcal{A}$ is constructed such that no environment machine or distinguisher $\mathcal{Z}$ can distinguish with non-negligible probability, its interactions with $\mathcal{A}$ and parties running the protocol $\Pi$ in the real world from its interaction with $\mathcal{A'}$ and ideal functionality $\mathcal{F}_{\sf AOT\mhyphen HAP}^{N \times 1}$ in the ideal world. The protocol $\Pi$ is said to be secure if the real world is as secure as the ideal world. Since the ideal world consists of incorruptible trusted parties, it is impossible for an adversary to attack the ideal world. Therefore, all adversarial attack must fail in real world too. Technically, for each case we show ${ \sf IDEAL}_{\mathcal{F}_{\sf AOT\mhyphen HAP}^{N \times 1}, \mathcal{A'}, \mathcal{Z}} \stackrel{c}{\approx}
{\sf REAL}_{\Pi, \mathcal{A}, \mathcal{Z}},$
where ${ \sf IDEAL}_{\mathcal{F}_{\sf AOT\mhyphen HAP}^{N \times 1}, \mathcal{A'}, \mathcal{Z}},
{\sf REAL}_{\Pi, \mathcal{A}, \mathcal{Z}}$ are as defined in section \ref{ucmodel}.

\noindent The security proof is presented using sequence of hybrid
games. Let ${\sf Pr}[{\sf Game}~ i]$ be the probability that
$\mathcal{Z}$ distinguishes the transcript of ${\sf Game} ~i$ from the real execution.

\noindent \textbf{(a) Security Against Malicious Receiver $R$}.
In this case, the real world adversary $\mathcal{A}$ corrupts the receiver $R$ before the execution of the protocol. The ideal world adversary $\mathcal{A'}$ simulates the honest sender $S$ by invoking a copy of $\mathcal{A}$ as follows.

\noindent $\bf{\underline{{\sf Game} ~0}}$: This game corresponds to the real world protocol
interaction in which the receiver interacts with the honest sender $S$.
So, ${\sf Pr}[{\sf Game}~ 0] = 0$.

\noindent $\bf{\underline{{\sf Game} ~1}}$: This game is the same
as ${\sf Game} ~0$ except that $\mathcal{A'}$ simulates ${\sf crs}$. The adversary $\mathcal{A'}$ simulates ${\sf crs} = ({\sf params}, {\sf GS}_S$, ${\sf GS}_R, \{\{A_{\ell, t}, B_{\ell, t}\}_{1 \leq t \leq n_{\ell}}\}_{1 \leq \ell \leq n})$ exactly as in the above game except that ${\sf GS}_S$ is in witness indistinguishability setting as explained below. The adversary $\mathcal{A'}$, first generates ${\sf params} = (p, \mathbb{G}_1, \mathbb{G}_2, \mathbb{G}_{T}, e, g_1, g_2) \leftarrow {\sf BilinearSetup}(1^{\rho})$ by invoking algorithm ${\sf BilinearSetup}$ discussed in section \ref{preliminaries}. The components $\{\{A_{\ell, t} = g_{1}^{a_{\ell,t}}, B_{\ell, t} = g_{2}^{a_{\ell,t}}\}_{1 \leq t \leq n_{\ell}}\}_{1 \leq \ell \leq n}$ are generated exactly as in ${\sf Game} ~0$. To generate ${\sf GS}_S$ in witness indistinguishability setting and ${\sf GS}_R$ in perfectly sound setting, the algorithm chooses $\xi_1, \xi_2, \widetilde{a}, \widetilde{b}, \widehat{\xi}_1, \widehat{\xi}_2, \widehat{a}, \widehat{b} \xleftarrow{\$} \mathbb{Z}_{p}$ and sets $u_1 = (g_{1}^{\widetilde{a}}, 1, g_{1})$, $u_2 = (1, g_{1}^{\widetilde{b}}, g_{1}), u_3 = u_{1}^{\xi_1}u_{2}^{\xi_2}(1, 1, g_{1}) = (g_{1}^{\widetilde{a}\xi_1}, g_{1}^{\widetilde{b}\xi_2}, g_{1}^{\xi_1+\xi_2+1})$, $v_1 = (g_{2}^{\widetilde{a}}, 1, g_{2}), v_2 = (1, g_{2}^{\widetilde{b}}, g_{2}), v_3 = v_{1}^{\xi_1}v_{2}^{\xi_2}(1, 1, g_{2}) = (g_{2}^{\widetilde{a}\xi_1}$, $g_{2}^{\widetilde{b}\xi_2}, g_{2}^{\xi_1+\xi_2+1})$,
$\widehat{u}_1 = (g_{1}^{\widehat{a}}, 1, g_{1}), \widehat{u}_2 = (1, g_{1}^{\widehat{b}}, g_{1}), \widehat{u}_3 = \widehat{u}_{1}^{\widehat{\xi}_1}\widehat{u}_{2}^{\widehat{\xi}_2}$ $= (g_{1}^{\widehat{a}\widehat{\xi}_1}, g_{1}^{\widehat{b}\widehat{\xi}_2}, g_{1}^{\widehat{\xi}_1+\widehat{\xi}_2})$, $\widehat{v}_1 = (g_{2}^{\widehat{a}}, 1, g_{2})$, $\widehat{v}_2 = (1, g_{2}^{\widehat{b}}, g_{2}), \widehat{v}_3 = \widehat{v}_{1}^{\widehat{\xi}_1}\widehat{v}_{2}^{\widehat{\xi}_2} = (g_{2}^{\widehat{a}\widehat{\xi}_1}, g_{2}^{\widehat{b}\widehat{\xi}_2}, g_{2}^{\widehat{\xi}_1+\widehat{\xi}_2})$, ${\sf GS}_S = (u_1, u_2, u_3, v_1, v_2, v_3)$, ${\sf GS}_R = (\widehat{u}_1, \widehat{u}_2, \widehat{u}_3, \widehat{v}_1, \widehat{v}_2, \widehat{v}_3)$, ${\sf \widetilde{t}_{sim}} = (\widetilde{a}, \widetilde{b}, \xi_1, \xi_2)$, ${\sf \widehat{t}_{ext}} = (\widehat{a}, \widehat{b}$, $\widehat{\xi}_1, \widehat{\xi}_2)$. The adversary $\mathcal{A'}$ distributes ${\sf crs} = ({\sf params}, {\sf GS}_S$, ${\sf GS}_R, \{\{A_{\ell, t}$, $B_{\ell, t}\}_{1 \leq t \leq n_{\ell}}\}_{1 \leq \ell \leq n})$ to all parties and keeps ${\sf \widetilde{t}_{sim}}$, ${\sf \widehat{t}_{ext}}$, $\{\{a_{\ell, t}\}_{1 \leq t \leq n_{\ell}}\}_{1 \leq \ell \leq n}$ hidden.

\noindent The common reference string ${\sf crs}$ simulated by $\mathcal{A'}$ and ${\sf crs}$ generated by ${\sf CrsSetup}$ in actual protocol run are computationally
indistinguishable by Theorem \ref{crs} in section \ref{noninteractive}. Therefore, there exists a negligible function
$\epsilon_{1}(\rho)$ such that $|{\sf Pr}[{\sf Game}~ 1] - {\sf
Pr}[{\sf Game}~ 0]| \leq \epsilon_{1}(\rho)$.

\noindent $\bf{\underline{{\sf Game} ~2}}$: This game is exactly
the same as ${\sf Game} ~1$ except that $\mathcal{A'}$ extracts attribute list $L = [v_{1, t_1}, v_{2, t_2}, \ldots, v_{n, t_n}]$ from $L' = \{R_{\ell}\}_{1\leq\ell \leq n}, ~\varphi$ sent by $\mathcal{A}$ in ${\sf Issue}$ phase as follows. The proof $\varphi = {\sf NIWI}_{{\sf GS}_R}\{(\{A_{\ell, t_{\ell}}, T_{\ell}\}_{1\leq\ell \leq n})~|~
e(g_1, \prod_{\ell=1}^{n}R_{\ell})$ $= e(\prod_{\ell=1}^{n}A_{\ell,t}, g_2)e(\prod_{\ell=1}^{n}T_{\ell}, g_2)\}$
was constructed by $\mathcal{A}$ which consists of commitments to $\{A_{\ell, t_{\ell}}, T_{\ell}\}_{1\leq\ell \leq n}$  and proof components of pairing equation $e(g_1, \prod_{\ell=1}^{n}R_{\ell}) = e(\prod_{\ell=1}^{n}A_{\ell,t}, g_2)$ $e(\prod_{\ell=1}^{n}T_{\ell}, g_2)$ using ${\sf GS}_R = (\widehat{u}_1, \widehat{u}_2, \widehat{u}_3, \widehat{v}_1, \widehat{v}_2$, $\widehat{v}_3)$ extracted from ${\sf crs}$, simulated by $\mathcal{A'}$ in ${\sf Game} ~1$. The adversary $\mathcal{A'}$ uses trapdoor ${\sf \widehat{t}_{ext}} = (\widehat{a}, \widehat{b}, \widehat{\xi}_1, \widehat{\xi}_2)$ to extract $\{A_{\ell, t_{\ell}}, T_{\ell}\}_{1\leq\ell \leq n}$ from their respective commitments as explained below.
For instance, let
${\sf Com}(A_{\ell, t_{\ell}}) = \mu_1(A_{\ell, t_{\ell}})\widehat{u}_{1}^{\widehat{r}_1}\widehat{u}_{2}^{\widehat{r}_2}\widehat{u}_{3}^{\widehat{r}_3} =
(X$ $= g_{1}^{\widehat{a}(\widehat{r}_{1} + \widehat{\xi}_{1}\widehat{r}_{3})}$, $Y = g_{1}^{\widehat{b}(\widehat{r}_{2} + \widehat{\xi}_{2}\widehat{r}_{3})},  Z = A_{\ell, t_{\ell}}\cdot g_{1}^{\widehat{r}_{1} + \widehat{r}_{2} + \widehat{r}_{3}(\widehat{\xi}_{1} + \widehat{\xi}_{2})})$, where $\widehat{r}_1, \widehat{r}_2, \widehat{r}_3$ are random values from $\mathbb{Z}_p$ which were used to generate ${\sf Com}(A_{\ell, t_{\ell}})$. The adversary $\mathcal{A'}$ computes
$$\frac{Z}{X^{\frac{1}{\widehat{a}}}Y^{\frac{1}{\widehat{b}}}} = \frac{A_{\ell, t_{\ell}}\cdot g_{1}^{\widehat{r}_{1} + \widehat{r}_{2} + \widehat{r}_{3}(\widehat{\xi}_{1} + \widehat{\xi}_{2})}}{(g_{1}^{\widehat{a}(\widehat{r}_{1} + \widehat{\xi}_{1}\widehat{r}_{3})})^{\frac{1}{\widehat{a}}}(g_{1}^{\widehat{b}(\widehat{r}_{2} + \widehat{\xi}_{2}\widehat{r}_{3})})^{\frac{1}{\widehat{b}}}} = A_{\ell, t_{\ell}}.$$
Similarly, $\mathcal{A'}$ extracts all the witnesses $\{A_{\ell, t_{\ell}}, T_{\ell}\}_{1 \leq \ell \leq n}$. Let $L$ be the set of attribute values which is initially empty. The adversary $\mathcal{A'}$ checks whether $A_{\zeta, \varsigma} = A_{\ell, t_{\ell}}$ for $\varsigma = 1$ to $n_{\zeta}$, $\zeta = 1$ to $n$. If yes, set $L = L \cup \{v_{\ell, t_{\ell}}\}$. In this way, $\mathcal{A'}$ recovers $L = [v_{1, t_1}, v_{2, t_2}, \ldots, v_{n, t_n}]$ and simulates $d_0$, $\{d'_{\ell, 1}, d_{\ell, 2}, d_{\ell, 3}\}_{1\leq\ell \leq n}$ exactly as in above game.

\noindent The randomized attribute secret key ${\sf ASK'}$ generated in ${\sf Game} ~2$ by $\mathcal{A'}$ is distributed identically to ${\sf ASK'}$ honestly generated by the sender $S$. Hence, $|{\sf Pr}[{\sf Game}~
2] - {\sf Pr}[{\sf Game}~ 1]| = 0$.

\noindent $\bf{\underline{{\sf Game} ~3}}$: This game is exactly
the same as ${\sf Game} ~2$ except that $\mathcal{A'}$ extracts the index $\sigma_j$ from the proof $\pi_{\sigma_j} = {\sf NIWI}_{{\sf GS}_R}\{(c_{\sigma_j}^{(1)}, \prod_{\ell=1}^{n}c_{\sigma_j, \ell}^{(4)}, g_{2}^{v_{\sigma_j}}, c_{\sigma_j}^{(2)})~|~
~e(\prod_{\ell=1}^{n}c_{\sigma_j, \ell}^{(4)}\cdot y,c_{\sigma_j}^{(1)}) = e(g_1, g_2) \wedge
e({\sf Req}_{\sigma_j}, g_2) = e(c_{\sigma_j}^{(2)}, g_2)e(B, g_{2}^{v_{\sigma_j}})\}$ by extracting witnesses ${\sf wit}_1 = c_{\sigma_j}^{(1)}$, ${\sf wit}_2 = \prod_{\ell=1}^{n}c_{\sigma_j, \ell}^{(4)}$, ${\sf wit}_3 = g_{2}^{v_{\sigma_j}}$, ${\sf wit}_4 = c_{\sigma_j}^{(2)}$ using ${\sf \widehat{t}_{ext}}$ exactly as $\mathcal{A'}$ has extracted $A_{\ell, t_{\ell}}$ in ${\sf Game} ~2$. The adversary $\mathcal{A'}$ checks whether $c_{i}^{(1)} = {\sf wit}_1$, $\prod_{\ell=1}^{n}c_{i, \ell}^{(4)} = {\sf wit}_2$, $c_{i}^{(2)} = {\sf wit}_4$ for $i =1, 2, \ldots, N$. Let $\sigma_j$ be the matching index. The adversary $\mathcal{A'}$ sends the message $({\sf transfer}, \sigma_j)$ to $\mathcal{F}_{\sf AOT\mhyphen HAP}^{N \times 1}$. The ideal functionality returns $m_{\sigma_j}$ to $\mathcal{A'}$ if $L$ satisfies the access policy $W_{\sigma_j}$ associated with $m_{\sigma_j}$, otherwise, $\bot$ is given to $\mathcal{A'}$. If no matching index $\sigma_j$ found and the proof $\pi_{\sigma_j}$ is correct, then this means that $\sigma_j \notin \{1, 2, \ldots, N\}$ and $\pi_{\sigma_j}$ is constructed by $\mathcal{A}$ for the ciphertext $\phi_{\sigma_j} = (c_{\sigma_j}^{(1)}, c_{\sigma_j}^{(2)}, c_{\sigma_j}^{(3)}, c_{\sigma_j}^{(4)}, c_{\sigma_j}^{(5)})$. Eventually, this shows that $\mathcal{A}$ is able to construct a valid BB signature $c_{\sigma_j}^{(1)}$,  thereby $\mathcal{A}$ outputs
$c_{\sigma_j}^{(1)}$ as a forgery contradicting the fact that the BB signature is
unforgeable assuming $q$-SDH problem
is hard \cite{boneh2004shortsignature}. Therefore, there exists a
negligible function $\epsilon_{3}(\rho)$ such that $|{\sf Pr}[{\sf
Game}~ 3] - {\sf Pr}[{\sf Game}~ 2]| \leq \epsilon_{3}(\rho)$.

\noindent $\bf{\underline{{\sf Game} ~4}}$: This game is the same as ${\sf Game}
~3$ except that  the response ${\sf Res}_{\sigma_j}$ and
proof $\delta_{\sigma_j}$ are simulated by $\mathcal{A'}$, for each transfer phase $j = 1, 2, \ldots, k$. To simulate response ${\sf Res}'_{\sigma_j}$,  $\mathcal{A'}$ first extracts $h^{v_{\sigma_j}}$ from its commitment $({\sf Com'}(h))^{v_{\sigma_j}}$ using ${\sf \widehat{t}_{ext}}$ as done in ${\sf Game}
~2$. The commitment to $h^{v_{\sigma_j}}$ is generated by $\mathcal{A}$ by raising power $v_{\sigma_j}$ to the commitment of $h$ given by $S$ in initialization phase as
\begin{align*}
({\sf Com'}(h))^{v_{\sigma_j}} &= (g_{2}^{\widehat{a}(\widehat{r}_1 + \widehat{r}_3\widehat{\xi}_1)}, g_{2}^{\widehat{b}(\widehat{r}_2 + \widehat{r}_3\widehat{\xi}_2)}, hg^{\widehat{r}_1+\widehat{r}_2+\widehat{r}_3(\widehat{\xi}_1+\widehat{\xi}_2)})^{v_{\sigma_j}}\\
&= (g_{2}^{\widehat{a}v_{\sigma_j}(\widehat{r}_1 + \widehat{r}_3\widehat{\xi}_1)}, g_{2}^{\widehat{b}v_{\sigma_j}(\widehat{r}_2 + \widehat{r}_3\widehat{\xi}_2)},
h^{v_{\sigma_j}}g^{v_{\sigma_j}(\widehat{r}_1+\widehat{r}_2+\widehat{r}_3(\widehat{\xi}_1+\widehat{\xi}_2))})\\
&= (g_{2}^{\widehat{a}(\widetilde{r}_1 + \widetilde{r}_3\widehat{\xi}_1)}, g_{2}^{\widehat{b}(\widetilde{r}_2 + \widetilde{r}_3\widehat{\xi}_2)}, h^{v_{\sigma_j}}g^{\widetilde{r}_1+\widetilde{r}_2+\widetilde{r}_3(\widehat{\xi}_1+\widehat{\xi}_2)}) = {\sf Com'}(h^{v_{\sigma_j}}),
\end{align*}
where $\widetilde{r}_1 = v_{\sigma_j}\widehat{r}_1, \widetilde{r}_2 = v_{\sigma_j}\widehat{r}_2, \widetilde{r}_3 = v_{\sigma_j}\widehat{r}_3$, $\widehat{r}_1, \widehat{r}_2, \widehat{r}_3 \xleftarrow{\$} \mathbb{Z}_{p}$. The adversary $\mathcal{A'}$ sets $\widetilde{c}_{\sigma_j, \ell}^{(5)} = c_{\sigma_j, \ell, t_{\ell}}^{(5)}$ if $v_{\ell, t_{\ell}} \in L$ for $\ell =1, 2, \ldots, n$ and
$${\sf Res}'_{\sigma_j} = \frac{c_{\sigma_j}^{(3)}\cdot \prod_{\ell=1}^{n} e(c_{\sigma_j, \ell}^{(4)}, d_{\ell, 1})\cdot e(B, h^{v_{\sigma_j}})}{e(c_{\sigma_j}^{(2)}, d_0)\cdot \prod_{\ell=1}^{n} e(\widetilde{c}_{\sigma_j, \ell}^{(5)}, d_{\ell, 3})\cdot{m_{\sigma_j}}},$$
%
%$${\sf Res}'_{\sigma_j} = \frac{c_{\sigma_j}^{(3)}\cdot e(B, h^{v_{\sigma_j}})}{m_{\sigma_j}\cdot e\left(\prod_{\ell=1}^{n}g_{1}^{s_{\sigma_j, \ell}}, g_{2}^{w}\right)},$$
where $L = [v_{1, t_1}, v_{2, t_2}, \ldots, v_{n, t_n}]$ and $m_{\sigma_j}$ are extracted in ${\sf Game} ~2$ and ${\sf Game} ~3$, respectively. The proof
\begin{gather*}
\delta_{\sigma_j} = {\sf NIZK}_{{\sf GS}_S}\{(a_{1, \sigma_j}, a_{2, \sigma_j}, a_{3, \sigma_j}, a_{4, \sigma_j})~|~
e({\sf Req}_{\sigma_j}^{-1}, a_{1, \sigma_j})e(a_{2, \sigma_j}, a_{3, \sigma_j}) = 1 \wedge\\
e(B^{-1}, a_{1, \sigma_j})e(a_{4, \sigma_j}, a_{3, \sigma_j}) = 1 \wedge
e(g_1, a_{3, \sigma_j}) = e(g_1, g_2)\},
\end{gather*}
where $a_{1, \sigma_j} = h$, $a_{2, \sigma_j} = {\sf Req}_{\sigma_j}^{\gamma}$, $a_{3, \sigma_j} = g_2$, $a_{4, \sigma_j} = B^{\gamma}$. To simulate the proof $\delta_{\sigma_j}$, $\mathcal{A'}$ sets $a_{1, \sigma_j} = a_{2, \sigma_j} = a_{3, \sigma_j} = a_{4, \sigma_j} = g_{2}^{0}$ and generate commitment to all these values. As $\mathcal{A'}$ knows the simulation trapdoor ${\sf \widetilde{t}_{sim}}$, it can open the commitment to any value of its choice as explained in Example \ref{exmp1} in section \ref{noninteractive}. The adversary $\mathcal{A'}$ opens commitment of $a_{3, \sigma_j}$ to $g_{2}^{0}$ in first and 2nd equation of $\delta_{\sigma_j}$ and to $g_{2}^{1}$ in 3rd equation of of $\delta_{\sigma_j}$. Thus, all the equations in $\delta_{\sigma_j}$ are simulated by $\mathcal{A'}$.

\noindent \emph{Claim 1. Under the DLIN assumption, the response
${\sf Res}_{\sigma_j}$ and the proof $\delta_{\sigma_j}$ honestly generated
by the sender $S$ are computationally indistinguishable from the
response ${\sf Res}'_{\sigma_j}$ and proof $\delta'_{\sigma_j}$ simulated by the simulator
$\mathcal{A'}$.}

\noindent \emph{Proof of Claim 1.} The simulated response ${\sf Res}'_{\sigma_j}$ is
\vspace{-.3cm}
\begin{eqnarray*}
{\sf Res}'_{\sigma_j} &=& \frac{c_{\sigma_j}^{(3)}\cdot \prod_{\ell=1}^{n} e(c_{\sigma_j, \ell}^{(4)}, d_{\ell, 1})\cdot e(B, h^{v_{\sigma_j}})}{e(c_{\sigma_j}^{(2)}, d_0)\cdot \prod_{\ell=1}^{n} e(\widetilde{c}_{\sigma_j, \ell}^{(5)}, d_{\ell, 3})\cdot{m_{\sigma_j}}}\\
&=& \frac{m_{\sigma_j}\cdot Y^{r_{\sigma_j}}\cdot e(B^{r_{\sigma_j}}, h)\cdot e(B, h^{v_{\sigma_j}})}{m_{\sigma_j}\cdot Y^{r_{\sigma_j}}}\\
%&=& \frac{m_{\sigma_j}\cdot e(g_1, g_{2}^{w})^{r_{\sigma_j}}\cdot e(B^{r_{\sigma_j}}, h)\cdot e(B, h^{v_{\sigma_j}})}{m_{\sigma_j}\cdot e\left(g_{1}^{r_{\sigma_j}}, g_{2}^{w}\right)} \\
&=& e(B, h)^{r_{\sigma_j} + v_{\sigma_j}},
\end{eqnarray*}
as $\frac{\prod_{\ell=1}^{n} e(c_{\sigma_j, \ell}^{(4)}, d_{\ell, 1})}{e(c_{\sigma_j}^{(2)}, d_0)\cdot \prod_{\ell=1}^{n} e(\widetilde{c}_{\sigma_j, \ell}^{(5)}, d_{\ell, 3})} = \frac{1}{Y^{r_{\sigma_j}}}$and, the honestly generated response ${\sf Res}_{\sigma_j}$ is
$${\sf Res}_{\sigma_j}  = e({\sf Req}_{\sigma_j}, h) = e(B^{r_{\sigma_j} + v_{\sigma_j}}, h) = e(B, h)^{r_{\sigma_j} + v_{\sigma_j}}.$$
The simulated response ${\sf Res}'_{\sigma_j}$ is distributed identically to honestly generated response ${\sf Res}_{\sigma_j}$.

\noindent As Groth-Sahai proofs are composable ${\sf NIZK}$  by Theorem \ref{niwi}, the simulated
proof $\delta'_{\sigma_j}$ is
computationally indistinguishable from the honestly generated
proof $\delta_{\sigma_j}$. Therefore,
we have $|{\sf Pr}[{\sf Game}~ 4] - {\sf Pr}[{\sf Game}~ 3]| \leq \epsilon_{4}(\rho)$, where $\epsilon_{4}(\rho)$
is a negligible function.

\noindent $\bf{\underline{{\sf Game} ~5}}$: This game is the same
as ${\sf Game} ~4$ except that $\mathcal{A'}$ replaces perfect database ${\sf DB} = \{(m_i, W_i)\}_{1 \leq i \leq N}$ by a random database ${\sf DB'} = \{(\widehat{m}_i, \widehat{W}_i)\}_{1 \leq i \leq N}$. The adversary $\mathcal{A'}$ generates $({\sf pk'},  \psi', {\sf cDB'})$, computes response ${\sf Res}_{\sigma_j}$, proof $\delta_{\sigma_j}$. But,
in this game the response ${\sf Res}_{\sigma_j}$ is computed on an invalid statement. The proof $\psi'$ is simulated exactly in the same way as $\mathcal{A'}$ simulates $\delta_{\sigma_j}$ in ${\sf Game} ~4$. In ${\sf Game} ~4$, ${\sf cDB}$ is the encryption of perfect database ${\sf DB}$ whereas in ${\sf Game} ~5$, ${\sf cDB'}$ is the encryption of random database ${\sf DB'}$. If the distinguisher $\mathcal{Z}$ can distinguish between the transcript of ${\sf Game} ~5$ from the transcript of ${\sf Game} ~4$, we can construct a solver for DLIN and DBDH using $\mathcal{Z}$ as a subroutine -- a contradiction as CP-ABE of \cite{nishide2008attribute} is semantically secure assuming the hardness of DLIN and DBDH problems.
Therefore, ${\sf Game}~ 4$ and ${\sf Game}~
5$ are computationally indistinguishable and $|{\sf
Pr}[{\sf Game}~ 5] - {\sf Pr}[{\sf Game}~ 4]| \leq
\epsilon_{5}(\rho)$, where $\epsilon_{5}(\rho)$ is a negligible
function.

\noindent Thus ${\sf Game} ~5$ is the ideal world interaction whereas ${\sf Game} ~0$ is the real
world interaction. Now\\
$
|{\sf Pr}[{\sf Game}~ 5] - [{\sf Game}~ 0]| \leq \sum_{t=1}^{5}|{\sf Pr}[{\sf Game}~ t] - [{\sf Game}~ (t-1)]|\leq   \epsilon_{6}(\rho),
$
where $\epsilon_{6}(\rho) =  \sum_{t=1}^{5}\epsilon_{t}(\rho)$ is a
negligible function. Hence, ${\sf IDEAL}_{\mathcal{F}_{\sf AOT\mhyphen HAP}^{N
\times 1}, \mathcal{A'}, \mathcal{Z}} \stackrel{c}{\approx} {\sf
REAL}_{\Pi, \mathcal{A}, \mathcal{Z}}$.

\noindent \textbf{(b) Security Against Malicious Sender $S$}.
In this case, $\mathcal{A}$ corrupts the sender $S$ and delivers all messages on behalf of $S$. The adversary $\mathcal{A'}$ simulates the actions of $R$ as follows.

\noindent $\bf{\underline{{\sf Game} ~0}}$: This game corresponds to the real world protocol
interaction in which $S$ communicates with honest $R$.
So, ${\sf Pr}[{\sf Game}~ 0] = 0$.

\noindent $\bf{\underline{{\sf Game} ~1}}$: This game is the same
as ${\sf Game} ~0$ except that $\mathcal{A'}$ simulates ${\sf crs} = ({\sf params}, {\sf GS}_S, {\sf GS}_R$, $\{\{A_{\ell, t}, B_{\ell, t}\}_{1 \leq t \leq n_{\ell}}\}_{1 \leq \ell \leq n})$. The adversary $\mathcal{A'}$, first generates ${\sf params} = (p, \mathbb{G}_1, \mathbb{G}_2, \mathbb{G}_{T}, e$, $g_1, g_2)$ by invoking algorithm ${\sf BilinearSetup}$ discussed in section \ref{preliminaries}. The components $\{A_{\ell, t} = g_{1}^{a_{\ell,t}}, B_{\ell, t} = g_{2}^{a_{\ell,t}}\}_{1 \leq t \leq n_{\ell}, 1 \leq \ell \leq n}$, are generated exactly as in ${\sf Game} ~0$. To generate ${\sf GS}_S$  and ${\sf GS}_R$ in perfectly sound setting, the algorithm, takes $\xi_1, \xi_2, \widetilde{a}, \widetilde{b}$, $\widehat{\xi}_1, \widehat{\xi}_2, \widehat{a}, \widehat{b}$ $\xleftarrow{\$} \mathbb{Z}_{p}$ and sets $u_1 = (g_{1}^{\widetilde{a}}, 1, g_1), u_2 = (1, g_{1}^{\widetilde{b}}, g_1), u_3 = (g_{1}^{\widetilde{a}\xi_1}$, $g_{1}^{\widetilde{b}\xi_2}, g_{1}^{\xi_1+\xi_2})$, $v_1 = (g_{2}^{\widetilde{a}}, 1, g_2), v_2 = (1, g_{2}^{\widetilde{b}}, g_2), v_3 = (g_{2}^{\widetilde{a}\xi_1}$, $g_{2}^{\widetilde{b}\xi_2}, g_{2}^{\xi_1+\xi_2})$,
$\widehat{u}_1 = (g_{1}^{\widehat{a}}, 1, g_1), \widehat{u}_2 = (1, g_{1}^{\widehat{b}}, g_1)$, $\widehat{u}_3 = (g_{1}^{\widehat{a}\widehat{\xi}_1}$, $g_{1}^{\widehat{b}\widehat{\xi}_2}, g_{1}^{\widehat{\xi}_1+\widehat{\xi}_2})$, $\widehat{v}_1 = (g_{2}^{\widehat{a}}, 1, g_2), \widehat{v}_2 = (1, g_{2}^{\widehat{b}}, g_2), \widehat{v}_3 = (g_{2}^{\widehat{a}\widehat{\xi}_1}$, $g_{2}^{\widehat{b}\widehat{\xi}_2}, g_{2}^{\widehat{\xi}_1+\widehat{\xi}_2})$, ${\sf GS}_S = (u_1, u_2$, $u_3, v_1, v_2, v_3), {\sf GS}_R = (\widehat{u}_1, \widehat{u}_2, \widehat{u}_3$, $\widehat{v}_1, \widehat{v}_2, \widehat{v}_3)$, ${\sf \widetilde{t}_{ext}} = (\widetilde{a}, \widetilde{b}, \xi_1, \xi_2)$, ${\sf \widehat{t}_{ext}} = (\widehat{a}, \widehat{b}, \widehat{\xi}_1, \widehat{\xi}_2)$. The adversary $\mathcal{A'}$ distributes ${\sf crs}$ to all parties and keeps ${\sf \widetilde{t}_{ext}}$, ${\sf \widehat{t}_{ext}}$, $\{\{a_{\ell, t}\}_{1 \leq t \leq n_{\ell}}\}_{1 \leq \ell \leq n}$ secret to itself.

\noindent The common reference string ${\sf crs}$ generated in ${\sf Game} ~1$ by $\mathcal{A'}$ is distributed identically to ${\sf crs}$ generated in ${\sf Game} ~0$ by ${\sf CrsSetup}$. Hence, $|{\sf Pr}[{\sf Game}~
1] - {\sf Pr}[{\sf Game}~ 0]| = 0$.

\noindent $\bf{\underline{{\sf Game} ~2}}$: This game is the same
as ${\sf Game} ~1$ except that $\mathcal{A'}$ upon receiving ${\sf pk}, \psi, {\sf cDB}$ from $\mathcal{A}$ extracts $h, g_{2}^{w}, g_{1}^{\alpha}, g_{2}^{\alpha}$ using ${\sf \widetilde{t}_{ext}}$ in the same way as $\mathcal{A'}$ extracts witnesses in ${\sf Game} ~2$ of Case (a). Using extracted $h, g_{2}^{w}, g_{2}^{\alpha}$, $\mathcal{A'}$ extracts $(m_i, W_i)$ from ciphertext $\phi_i$ for each $i =1, 2, \ldots, N$ as follows. The adversary $\mathcal{A'}$ extracts $c_{i}^{(2)} = B^{r_i}$, $ c_{i}^{(3)} = m_i\cdot Y^{r_i}\cdot e(c_{i}^{(2)}, h)$, $c_{i}^{(4)} = \{c_{i, \ell}^{(4)}\}_{1\leq \ell \leq n} =\{g_{1}^{s_{i, \ell}}\}_{1\leq \ell \leq n}$, $c_{i}^{(5)} = \{\{c_{i, \ell, t}^{(5)}\}_{1\leq t\leq n_{\ell}}\}_{1\leq \ell \leq n}$ from $\phi_i = (c_{i}^{(1)}, c_{i}^{(2)}, c_{i}^{(3)}, c_{i}^{(4)}, c_{i}^{(5)})$ and computes
\vspace{-.3cm}
\begin{eqnarray*}
\frac{c_{i}^{(3)}}{e\left(\prod_{\ell=1}^{n}c_{i, \ell}^{(4)}, g_{2}^{w}\right)e(c_{i}^{(2)}, h)} = \frac{m_i\cdot Y^{r_i}e(B^{r_i}, h)}{e\left(\prod_{\ell=1}^{n}g_{1}^{s_{i, \ell}}, g_{2}^{w}\right)e(B^{r_i}, h)}\\
= \frac{m_i\cdot Y^{r_i}e(B^{r_i}, h)}{e\left(g_{1}^{\sum_{\ell=1}^{n}s_{i, \ell}}, g_{2}^{w}\right)e(B^{r_i}, h)} = m_i
\end{eqnarray*}
as $\sum_{\ell=1}^{n}s_{i, \ell} = r_i, Y = e(g_1, g_{2}^{w})$. Let $W_i = [W_{i, 1}, W_{i, 2}$, $\ldots, W_{i, n}]$ be the access policy associated with $m_i$ which is initially empty. The adversary $\mathcal{A'}$ checks whether
$e((c_{i, \ell, t}^{(5)})^{\frac{1}{a_{\ell, t}}}$, $g_2)$ $= e(c_{i, \ell}^{(4)}, g_{2}^{\alpha})$, if yes, $W_{i, \ell} = W_{i, \ell}\cup \{v_{\ell, t}\}$ for $t = 1, 2, \ldots, n_{\ell}$, for $\ell=1, 2, \ldots, n$. In this way, $W_i$ is recovered by $\mathcal{A'}$. The adversary $\mathcal{A'}$ sends $({\sf initdb}, {\sf DB} = \{m_i, W_i\}_{1 \leq i\leq N})$ to $\mathcal{F}_{\sf AOT\mhyphen HAP}^{N \times 1}$. Upon receiving message $({\sf issue}, L)$, $\mathcal{A'}$ simulates $R$'s  side of the ${\sf Issue}$ phase. The adversary $\mathcal{A'}$ picks attributes, constructs attribute list $\widetilde{L} = [v_{1, t_1}, v_{2, t_2}, \ldots, v_{n, t_n}]$ such that $\widetilde{L}\models W_1$
and simulates $R_{\ell}$ by taking $B_{\ell, t_{\ell}}$, from ${\sf crs}$, corresponding to $v_{\ell, t_{\ell}}$ from $\widetilde{L}$ and sets $L'= \{R_{\ell}\}_{1\leq \ell \leq n}$, where $R_{\ell} = B_{\ell, t_{\ell}}\cdot g_{2}^{z_{\ell}}$. The proof $\varphi$ is simulated with the witnesses $\{A_{\ell, t_{\ell}}, T_{\ell}\}_{1\leq \ell \leq n}$, where $T_{\ell} = g_{1}^{z_{\ell}}$. The proof $\varphi$ and $L'$ is given to $\mathcal{A}$ who in turn gives randomized attribute key ${\sf ASK'} = (d_0 = g_{2}^{\frac{w+s}{\beta}}, \{d'_{\ell, 1} = g_{2}^{s}\cdot R_{\ell}^{\lambda_{\ell}}, d_{\ell, 2} = g_{2}^{\lambda_{\ell}}, d_{\ell, 3} = g_{2}^{\frac{\lambda_{\ell}}{\alpha}}\}_{1\leq\ell \leq n})$. The adversary $\mathcal{A'}$ checks whether $\frac{e(g_1, d'_{\ell, 1}\cdot d_{\ell, 2}^{-z_{\ell}})e(g_{1}^{\alpha}, d_{\ell, 3})}{e(B, d_0)e(g_{1}^{-1}, g_{2}^{w})} = e(A_{\ell, t_{\ell}}\cdot g_1, d_{\ell, 2})$ for $\ell=1, 2, \ldots, n$. If yes, a bit $b=1$ to given to $\mathcal{F}_{\sf AOT\mhyphen HAP}^{N \times 1}$ by $\mathcal{A'}$. Otherwise, $\mathcal{A'}$ gives $b=0$ to $\mathcal{F}_{\sf AOT\mhyphen HAP}^{N \times 1}$.

\noindent As Groth-Sahai proofs are composable ${\sf NIWI}$  by Theorem \ref{niwi}, the simulated $L', \varphi$ are computationally indistinguishable from the honestly generated $L', \varphi$. Therefore,
we have $|{\sf Pr}[{\sf Game}~ 2] - {\sf Pr}[{\sf Game}~ 1]| \leq \epsilon_{1}(\rho)$, where $\epsilon_{1}(\rho)$
is a negligible function.

\noindent $\bf{\underline{{\sf Game} ~3}}$: This game is the same
as ${\sf Game} ~2$ except that $\mathcal{A'}$ upon receiving the message $({\sf transfer}, {\sf \sigma}_j)$, runs the receiver $R$'s side of the ${\sf Transfer}$ phase. The adversary $\mathcal{A'}$ simulates ${\sf Req}_{\sigma_j}, \pi_{\sigma_j}$ as follows. The adversary $\mathcal{A'}$ replaces ${\sf Req}_{\sigma_j}$ by ${\sf Req}_{1} = c_{1}^{(2)}\cdot B^{v_{1}}$ and proof $\pi_{\sigma_j}$ by $\pi_{1} = {\sf NIWI}_{{\sf GS}_R}\{(c_{1}^{(1)}$, $\prod_{\ell=1}^{n}c_{1, \ell}^{(4)}, g_{2}^{v_{1}}, c_{1}^{(2)})~|~$
$e(\prod_{\ell=1}^{n}c_{1, \ell}^{(4)}\cdot y, c_{1}^{(1)}) = e(g_1, g_2) \wedge
e({\sf Req}_{1}, g_2)$ $= e(c_{1}^{(2)}, g_2)e(B, g_{2}^{v_{1}})\}$, ${\sf Pri}_{1} = H^{v_{1}}$. The adversary $\mathcal{A'}$ gives $\pi_1$ and ${\sf Req}_{1}$ to $\mathcal{A}$. Upon receiving response ${\sf Res}_1, \delta_1$ from $\mathcal{A}$, $\mathcal{A'}$ checks whether $\frac{c_{1}^{(3)}\cdot \prod_{\ell=1}^{n} e(c_{1, \ell}^{(4)}, d_{\ell, 1})\cdot H^{v_{1}}}{e(c_{1}^{(2)}, d_0)\cdot \prod_{\ell=1}^{n} e(\widetilde{c}_{1, \ell}^{(5)}, d_{\ell, 3})\cdot{\sf Res}_{1}} = m_{1}$. If yes, a bit $b=1$ to given to $\mathcal{F}_{\sf AOT\mhyphen HAP}^{N \times 1}$ by $\mathcal{A'}$. Otherwise, $\mathcal{A'}$ gives $b=0$ to $\mathcal{F}_{\sf AOT\mhyphen HAP}^{N \times 1}$.
As Groth-Sahai proofs are composable ${\sf NIWI}$  by Theorem \ref{niwi}, the simulated request ${\sf Req}_{1}$ is computationally indistinguishable from the honestly generated ${\sf Req}_{\sigma_j}$. Therefore,
we have $|{\sf Pr}[{\sf Game}~ 3] - {\sf Pr}[{\sf Game}~ 2]| \leq \epsilon_{2}(\rho)$, where $\epsilon_{2}(\rho)$
is a negligible function

\noindent Thus ${\sf Game} ~3$ is the ideal world interaction whereas ${\sf Game} ~0$ is the real
world interaction. Now\\
$
|{\sf Pr}[{\sf Game}~ 3] - [{\sf Game}~ 0]| =
\sum_{t=1}^{3}|{\sf Pr}[{\sf Game}~ t] - [{\sf Game}~ (t-1)]|,
$
where $\epsilon_{3}(\rho) =
\epsilon_{2}(\rho) + \epsilon_{1}(\rho) + 0$ is a
negligible function. Hence, ${\sf IDEAL}_{\mathcal{F}_{\sf AOT\mhyphen HAP}^{N
\times 1}, \mathcal{A'}, \mathcal{Z}} \stackrel{c}{\approx} {\sf
REAL}_{\Pi, \mathcal{A}, \mathcal{Z}}$.
\end{proof}

\section{Comparison}
We compare our proposed issuer-free adaptive oblivious transfer with hidden access policy (${\sf AOT\mhyphen HAP}$) with \cite{abe2013universally}, \cite{camenisch2012oblivious}, \cite{camenisch2011oblivious}, \cite{guleria2014adaptive}.
The subtle differences between our scheme and \cite{abe2013universally}, \cite{camenisch2012oblivious}, \cite{camenisch2011oblivious}, \cite{guleria2014adaptive} are listed below.
\begin{enumerate}
\item The schemes of \cite{camenisch2012oblivious}, \cite{camenisch2011oblivious} and \cite{guleria2014adaptive}, to the best of our knowledge, are the only existing ${\sf AOT\mhyphen HAP}$, but they are not secure in UC framework.
\item Abe {\em et al.} \cite{abe2013universally} introduced the first adaptive oblivious transfer with access policy in UC framework, but access policies are not hidden.
\item  Guleria and Dutta \cite{guleria2015universally} presented the concept of issuer-free adaptive oblivious transfer with public access policies.
\item  The schemes \cite{abe2013universally}, \cite{camenisch2012oblivious}, \cite{camenisch2011oblivious}, \cite{guleria2014adaptive} assume an issuer apart from a sender and multiple receivers in their constructions and are secure under the restriction that the issuer never colludes with a set of receivers.
\end{enumerate}

\begin{table}[htbp]
\scriptsize
\centering
\begin{tabular}{|c|c|c|c|c|c|c|}
\hline
UC Secure & \multicolumn{2}{|c|}{Pairing}&\multicolumn{3}{|c|}{Exponentiation}& Hidden \\
Schemes&\multicolumn{2}{c|}{${\sf PO}$}&\multicolumn{3}{c|}{${\sf EXP}$}& ${\sf AP}$\\
\cline{2-7}
 & ${\sf Transfer}$ & ${\sf DBSetup}$ & ${\sf Transfer}$ & ${\sf DBSetup}$ &
 ${\sf CRSG}$&   \\
\cline{1-7}
 %\cite{green2008universally}  & $\geq 207k$& $24N + 1$ & $249k$ &$20N + 13$ & 18&-- \\
% \hline
% \cite{rial2009universally}& $> 450k$ & $15N + 1$ & $223k$ &$12N + 9$ & 15 & -- \\
% \hline
% \cite{guleria2014efficient} & $147k$ & $5N+1$ & $150k $ & $17N+5$& 18& --\\
% \hline
 \cite{abe2013universally} & $(100n + 199)k$ & $(n+21)N$ & $(138n+237)k$ &$(2n+21)N + 2n+20$ & $n+26$& $\times$\\
 \hline
 \cite{guleria2015universally}& $2\displaystyle\sum_{j=1}^{k}n_{\sigma_j}+86k$ & $N + 2$ & $96k$ &$3\displaystyle\sum_{i=1}^{N}n_i+3N + 59$ & 10& $\times$\\
 \hline
 Ours & $(2n+149)k$ & $N + 2$ & $112k$ &$\leq((m+n+3)N + 112)$ & $2m+20$ &$\surd$\\
 \hline
\end{tabular}
\vspace{2mm}
 \caption{Comparison Summary of computation cost in $k$ transfer phases and initialization phase
 (${\sf PO}$ stands for number of pairing operations, ${\sf EXP}$ for number of
 exponentiation operations,
 ${\sf CRSG}$ for ${\sf crs}$ generation, and ${\sf AP}$ for access control, $n$ is the number of attributes, $m$ is total number of values which $n$ attributes can take, $N$ is the database size).}
\label{tab6}
\end{table}

\begin{table}[htbp]
\scriptsize
\centering
\begin{tabular}{|c|c|c|c|c|c|}
\hline
UC Secure & Communication &\multicolumn{2}{|c|}{Storage}& Security Assumptions& Issuer-Free \\
Schemes & &  \multicolumn{2}{c|}{} & &\\
\cline{2-4}
&Request + Response & ${\sf crs}$-Size & (${\sf cDB}$ + ${\sf pk}$)Size &-- &\\
\cline{1-6}
 %\cite{green2008universally}  & $144k$ & $14$  & $18N + 11$& SXDH, $q$-LRSW, DLIN &--\\
% \hline
% \cite{rial2009universally}& $93k$ & $23$ & $12N + 7$& HSDH, TDH, DLIN& --\\
% \hline
% \cite{guleria2014efficient} & $75k$ & 16&$12N+5$& $q$-SDH, DLIN&--\\
% \hline
 \cite{abe2013universally} & $(125+64n)k$& $n+28$ & $16N + n + 20$& SXDH, XDLIN& $\times$\\
 \hline
 \cite{guleria2015universally}& $47\mathbb{G}k$& $11\mathbb{G}$ & $\left(\displaystyle\sum_{i=1}^{N}n_i+3N + m+11\right)\mathbb{G}$&DLIN, $q$-SDH , $q$-DBDHE &$\surd$\\
 \hline
 Ours & $62k$& $2m+22$ & $\leq((m+n+3)N + 7)$&$q$-SDH, DBDH, DLIN&$\surd$ \\
 \hline
\end{tabular}
\vspace{2mm}
 \caption{Comparison summary of communication cost in $k$ transfer phases and initialization phase (${\sf cDB}$ stands for ciphertext database, ${\sf pk}$ for public key, $n$ is the number of attributes, $m$ is total number of values which $n$ attributes can take, $N$ is the database size, SXDH - symmetric external Diffie-Hellman assumption, XDLIN -external decision Linear assumption, $q$-SDH - $q$-strong Diffie-Hellman assumption, DBDH-decision bilinear Diffie-Hellman assumption, DBDHE-decision bilinear Diffie-Hellman exponent assumption, DLIN-decision Linear assumption).}
\label{tab2}
\end{table}

\noindent In contrast to \cite{abe2013universally}, \cite{camenisch2012oblivious}, \cite{camenisch2011oblivious}, \cite{guleria2014adaptive}, we introduce the \emph{first} issuer-free ${\sf AOT\mhyphen HAP}$ in UC framwork. Our ${\sf AOT\mhyphen HAP}$ is issuer-free, realizes hidden access policy and achieves UC security.  Table \ref{tab6} provides the computation cost involved in $k$ ${\sf Transfer}$ phases, in algorithms ${\sf DBSetup}$ and ${\sf CRSSetup}$. In addition, ${\sf issue}$ phase takes $20n+8$ exponentiations and $2n+15$ pairing. Table \ref{tab2} exhibits (i) communication cost in $k$ transfer phases, (ii) storage complexity of common reference string ${\sf crs}$, (iii)
storage complexity of ciphertext database and public key size and (iv) complexity assumption. As illustrated in Table \ref{tab6} and \ref{tab2}, our protocol is more efficient as compared to the only existing UC secure adaptive oblivious without hidden access policy.

\section{Conclusion}
We have proposed the first issuer-free adaptive oblivious transfer with hidden access policy (${\sf AOT\mhyphen HAP}$) in universal composable (UC) framework. In issuer-free ${\sf AOT\mhyphen HAP}$, the sender publishes the ciphertext database encrypted under associated access policies. The receiver interacts with the sender in order to decrypt the messages of its choice without revealing its identity, attribute set and choice of messages. The receiver either recovers the correct message or a garbage one depending upon whether the receiver's attribute set satisfies the access policy associated with the message.
The proposed ${\sf AOT\mhyphen HAP}$ has been proved UC secure under the hardness of $q$-strong Diffie-Hellman (SDH), decision Linear (DLIN) and decision bilinear Diffie-Hellman (DBDH) problems in the presence of malicious adversary. Moreover, the protocol is efficient as compared to the existing similar schemes.

%\bibliographystyle{splncs03}
%\bibliography{paper9}

\end{document}